%% file: main.tex
\documentclass[english,11pt]{article}
\usepackage[T1]{fontenc}

\usepackage[latin9]{inputenc}
\usepackage{geometry}
\usepackage{amsmath}
\usepackage{amsthm}
\usepackage{amssymb}
\usepackage{fancybox}
\usepackage{float}
\usepackage[normalem]{ulem}
\usepackage{comment}
\usepackage{mathtools}
\usepackage{todonotes}

\makeatletter

\numberwithin{equation}{section}
\numberwithin{figure}{section}
\theoremstyle{plain}
	\newtheorem{theorem}{Theorem}[section]

	\newtheorem{lemma}[theorem]{Lemma}

	\newtheorem{corollary}[theorem]{Corollary}

	\newtheorem{fact}[theorem]{Fact}
	\newtheorem{obs}[theorem]{Observation}

\theoremstyle{definition}
	\newtheorem{definition}[theorem]{Definition}
	\newtheorem{defn}[theorem]{Definition}
	\newtheorem*{remark*}{Remark}
 
\pdfoutput=1

\newif\ifusebb
\usebbtrue
\usebbfalse
\ifusebb
  \usepackage{bbm}
  \DeclareSymbolFont{bbold}{U}{bbold}{m}{n}
  \DeclareSymbolFontAlphabet{\mathbbold}{bbold}
  \newcommand{\allones}{\ensuremath{\mathbbm{1}}}
  \newcommand{\allzeros}{\ensuremath{\mathbbold{0}}}
\else
  \newcommand{\allones}{\vec{1}}
  \newcommand{\allzeros}{\vec{0}}
\fi
 
 \renewcommand{\epsilon}{\varepsilon}

\usepackage[pdftex,bookmarks,colorlinks]{hyperref}
\usepackage{enumitem}
\setlist[description,1]{align=left,leftmargin=0.1in} 
\setlist[itemize,1]{leftmargin=*}
\usepackage{color,soul}
\usepackage{tikz,wrapfig}

\usepackage[lined,boxed,ruled,norelsize,algo2e]{algorithm2e}

\hypersetup{colorlinks=true,citecolor=green}

\makeatother

\usepackage{babel}

\newenvironment{fminipage}%
  {\begin{Sbox}\begin{minipage}}%
  {\end{minipage}\end{Sbox}\fbox{\TheSbox}}

\def\pleq{\preccurlyeq}
\def\pgeq{\succcurlyeq}

\def\defeq{\stackrel{\mathrm{def}}{=}}

\def\abs#1{\left|#1  \right|}

\def\norm#1{\left\| #1 \right\|}

\renewcommand\AA{\mvar{A}}
\newcommand\BB{\mvar{B}}
\newcommand\DD{\mvar{D}}
\newcommand\EE{\mvar{E}}
\newcommand\HH{\mvar{H}}
\newcommand\II{\mvar{I}}
\newcommand\MM{\mvar{M}}
\newcommand\LL{\mvar{L}}
\newcommand\PP{\mvar{P}}
\newcommand\PPtil{\mvar{\widetilde P}}
\newcommand\UU{\mvar{U}}
\newcommand\VV{\mvar{V}}
\newcommand\XX{\mvar{X}}
\newcommand\YY{\mvar{Y}}

\newcommand\Otil{\widetilde{O}}

\newcommand{\oracle}{\mathcal{O}}

\global\long\def\R{\mathbb{{R}}}

\global\long\def\boldVar#1{\mathbf{#1}}
\global\long\def\mvar#1{\boldVar{#1}}
\global\long\def\vvar#1{\vec{#1}}

\global\long\def\defeq{\stackrel{\mathrm{{\scriptscriptstyle def}}}{=}}

\global\long\def\norm#1{\|#1\|}

\newcommand{\ve}{\vvar{e}}

\newcommand{\vindic}[1]{\ve_{i}}

\newcommand{\mI}{\mvar I}

\newcommand{\ms}{\mvar S}

\newcommand{\eps}{\epsilon}

\newcommand{\mzero}{\mvar 0}

\renewcommand{\preceq}{\pleq}
\renewcommand{\succeq}{\pgeq}

\newcommand{\mdiag}{\mathbb{D}}
\newcommand{\lA}{\ell_{\AA}}
\newcommand{\uA}{s_{\AA}}
\newcommand{\wA}{w_{\AA}}
\newcommand{\SDDM}{SDD}

\global\long\def\oracle{\mathcal{O}}

\global\long\def\abs#1{\left|#1\right|}

\global\long\def\tr{\mathrm{tr}}

\global\long\def\argmin{\mathrm{argmin}}

\global\long\def\supp{\mathrm{supp}}

\renewcommand{\preceq}{\pleq}

\newcommand{\Sc}{\operatorname{Sc}}

\renewcommand{\ln}{\log}

\begin{document}

\title{
Matrix Scaling and Balancing via Box Constrained Newton's Method and Interior Point Methods
}

\author{Michael B. Cohen\thanks{This material is based upon work supported by the National Science
Foundation under Grant No. 1111109 and Grant No. 1553428, and by the
National Defense Science and Engineering Graduate Fellowship.}\\
MIT\\
micohen@mit.edu
\and 
Aleksander M\k{a}dry\thanks{This material is based upon work supported by the National Science
Foundation under Grant No. 1553428.} \\
MIT\\
madry@mit.edu
\and
Dimitris Tsipras\footnotemark[2] \\
MIT\\
tsipras@mit.edu
\and  Adrian Vladu\thanks{This material is based upon work supported by the National Science
Foundation under Grant No. 1111109 and Grant No. 1553428} \\
MIT\\
avladu@mit.edu}

\date{}
\maketitle
\thispagestyle{empty}
\setcounter{page}{0}

\begin{abstract}
	\normalsize
In this paper, we study matrix scaling and balancing, which are fundamental problems in scientific computing, with a long line of work on them that dates back to the 1960s. We provide algorithms for both these problems that, ignoring logarithmic factors involving the dimension of the input matrix and the size of its entries, both run in time $\widetilde{O}\left(m\log \kappa \log^2 (1/\epsilon)\right)$  where $\epsilon$ is the amount of error we are willing to tolerate.
Here, $\kappa$ represents the ratio between the largest and the smallest entries of the optimal scalings. This implies that our algorithms run in nearly-linear time whenever $\kappa$ is quasi-polynomial, which includes, in particular, the case of strictly positive matrices. We complement our results by providing a separate algorithm that uses an interior-point method and runs in time $\widetilde{O}(m^{3/2} \log (1/\epsilon))$.

In order to establish these results, we develop a new second-order optimization framework that enables us to treat both problems in a unified and principled manner. This framework identifies a certain generalization of linear system solving that we can use to efficiently minimize a broad class of functions, which we call \textit{second-order robust}. We then show that in the context of the specific functions capturing matrix scaling and balancing, we can leverage and generalize the work on Laplacian system solving to make the algorithms obtained via this framework very efficient.
\end{abstract}

\vfill

\pagebreak{}

\input{intro}
\input{approach}
\input{prelim}
\input{newton}
\input{bal}
\input{schur}
\input{ipm}

\appendix
\input{appendix.tex}

\bibliographystyle{plain}
\bibliography{ref}

\end{document}

%% file: intro.tex

\section{Introduction}
Matrix balancing and scaling are problems of fundamental importance in scientific computing, as well as in statistics, operations research, image reconstruction, and engineering. The literature on these problems \cite{Osb60, PR69, grad1971matrix, hartfiel1971concerning, eaves1985line, KKS97, schneider1991max, sink67, sinkhorn1964relationship, wilkinson1994rounding, raghavan1984pairs, brown1959note, friedland1988additive, idel2016review} is truly extensive and dates back to 
1960s. They both are key primitives in most mainstream numerical software packages (MATLAB, R, LAPACK, EISPACK)~\cite{matlabeig, matlabbalance,rbalance,rexpm,lapackbalance}. Also, both these problems can be seen as task in which we are aiming to find diagonal scalings of a given matrix so that the rescaled matrix gains some favorable structure. 

More specifically, in the matrix scaling problem, we are given a nonnegative matrix $\AA$, and our goal is to find diagonal matrices $\XX,\YY$ such that the matrix $\XX\AA\YY$ has prescribed row and column sums. The most common instance of this problem is the one where we want to scale the matrix so to make it doubly stochastic -- in other words, we want to make all row and column sums be equal to one. This procedure has been repeatedly used since as early as 1937 in a number of diverse areas, such as telecommunication~\cite{kruithof1937ingenieur}, engineering~\cite{brown1959note}, statistics~\cite{fienberg1976analysis, sinkhorn1964relationship}, machine learning~\cite{cuturi2013sinkhorn}, and even computational complexity~\cite{LSW98,Gurvits03}.  A standard application for scaling is preconditioning linear system solving. Given a linear system $\AA x = b$, one can produce a solution by computing  $\YY (\XX\AA\YY)^{-1} \XX b$, since applying the inverse of $\XX\AA\YY$ is more numerically stable procedure than directly applying the inverse of $\AA$~\cite{wilkinson1994rounding}. Another example application, which commonly occurs in statistics, is iterative proportional fitting. This primitive is often used for standardizing cross-tabulations and has been studied since 1912~\cite{yule1912methods}. Even more interestingly, matrix scaling turned out to have surprising connections to fundamental problems in the theory of computation. Notably, in~\cite{LSW98}, it is observed that scaling can be used to approximate the permanent of any nonnegative matrix within a multiplicative factor of $e^n$. Furthermore, deciding whether the permanent of a bipartite graph's adjacency matrix is $0$ or at least $1$ is equivalent to deciding whether that graph contains a perfect matching. Such scaling--based method can, as a matter of fact, be used to compute maximum matchings in bipartite graphs, which is a classic and intensely studied problem in graph algorithms~\cite{edmonds1965maximum, gabow1989faster, madry2013navigating}. For more history and information on this problem, we refer the reader to Idel's extensive survey~\cite{idel2016review}, or~\cite{schneider1990comparative} for a list of applications.

Now, in the matrix balancing problem, we are, again, given a nonnegative matrix $\AA$, and our goal here is to find a diagonal matrix $\DD$ such that the matrix $\DD\AA\DD^{-1}$ is \emph{balanced}, that is the sum of each row is equal to the sum of the corresponding column. 
This procedure has been introduced first by Osborne~\cite{Osb60}, who was using it to precondition matrices in order to increase the accuracy of the eigenvalue computation. (Note that the balancing operation does not change the eigenvalues of the matrix.) The initially proposed algorithm for it was based on a simple iterative approach, and was then followed by a long series of improvements and extensions. The initial work on this problem focused on a variant in which one aims to balance $\ell_2$-norms of rows and columns. It turns out, however, that the $\ell_1$-norm--based version we study here is equivalent. In fact, balancing problems with respect to $\ell_p$ norms, with constant $p \geq 1$, are all reducible to each other.

\subsection{Previous Work}
The early methods used for solving these problems -- Osborne's iteration for balancing, and the RAS method for scaling -- are simple iterative algorithms. However, merely the task of analyzing their convergence turned out to be a major challenge. Significant effort has gone into understanding their convergence~\cite{chen2000balancing, schulman2015analysis,  ORY17, LSW98}, and providing better analyses or better iterative methods resulted in a long line of work in this context.

The major shortcoming of the methods obtained so far for exactly solving the problem (depending only logarithmically on $1/\eps$)  is their
very large running time. In the following discussion we ommit runtime factors that depend (logarithmically) on the size of the input entries. For matrix scaling, Kalantari and Kachiyan~\cite{KK96}
obtained an algorithm that finds an $\eps$-approximate solution and runs in
time $\Otil(n^4 \log (1/\eps))$, where $n$ denotes the dimension of the matrix (we can assume the matrix is
square w.l.o.g.) and $\eps$ is the desired accuracy parameter\footnote{The precise definition of $\eps$ varies across papers. However, in the regime of logarithmic running time dependence on $1/\eps$ we are interested in here, all these definitions are essentially equivalent.}. This algorithm was based on the ellipsoid method. These authors also proposed -- but not formally analyzed -- an algorithm based on
interior point method, which they expected to run in time $\Otil(m^{3.5} \log
(1/\eps))$, where $m$ denotes the number of non-zero entries of the input matrix. Then, Nemirovsky and Rothblum~\cite{nemirovski1999complexity}
analyzed an interior point method--based algorithm which run in time $\Otil(m^4 \log
(1/\eps))$. Finally, Linial, Samorodnitsky, and Wigderson~\cite{LSW98} gave
an $\Otil(n^7\log(1/\eps))$ time algorithm that is also strongly polynomial, in the sense that it does not depend at all on the size of input entries.

For the case of matrix balancing, Parlett and Reinsch~\cite{PR69} provided an iterative method based on Osborne's iteration, without proving convergence. Then, Grad~\cite{grad1971matrix} proved that Osborne's iteration converges in the limit. The first polynomial time bound was obtained by Kalantari, Khachiyan, and Shokoufandeh~\cite{KKS97}, who gave an algorithm with running time  $\Otil(n^4 \log (1/\eps))$. 

Alternatively, if one is interested in the regime where the running time is allowed to depend polynomially -- instead of logarithmically -- on the (inverse of the) desired accuracy of the solution, there are algorithms that have an even better dependence on the other parameters. Specifically, the current state-of-the-art is given by Linial, Samorodnitsky, and Wigderson~\cite{LSW98}, who obtain $O(n^3\eps^{-2})$ running time for the scaling problem. In the case of the balancing problem, recently, Ostrovsky, Rabani, and Yousefi~\cite{ORY17} made a significant progress by obtaining running times of  $\Otil(m+n \eps^{-2})$ and $\Otil(n^{3.5}\eps^{-1})$. 

Finally, another important line of work in this domain was focused on the related $\ell_\infty$ variant of the balancing problem, where the maximum entry of each row is required to be equal to the maximum entry of the corresponding column. Schneider and Schneider~\cite{schneider1991max} gave a non-iterative algorithm running in time $O(n^4)$, improved to  $\Otil(mn + n^2)$ by Young, Tarjan, and Orlin~\cite{young1991faster}. More recently, Schulman and Sinclair~\cite{schulman2015analysis} provided an analysis of the classical Osborne-Parlett-Reinsch obtaining a running time of $\Otil(n^2m)$, and gave a version of it with running time $\Otil(n^3 \log (1/\eps))$.

\subsection{Our Contributions}
We provide algorithms for both matrix scaling and balancing problems.

For the matrix scaling problem, we establish an algorithm that runs in time
$$\Otil\left(  m \log (\kappa(\UU^*) + \kappa(\VV^*)) \log^2 \frac{\uA}{\eps}   \right)\,{,}$$
where $\UU^*$ and $\VV^*$ are the optimal scaling matrices, $\kappa(\cdot)$ is the
maximum ratio between the diagonal entries of its argument, $\uA$ is the sum of the entries in the input matrix, and $\eps$ is the measure of the target error of the scaling, formally defined in Definition~\ref{def:eps_scaling}.

For the matrix balancing problem, we establish a running time of 
$$\Otil\left(m\log \kappa(\DD^* )\log^2 \frac{w_\AA}{\eps}\right)\,{,}$$ 
where $w_\AA$ is the ratio of the sum of
the entries to the minimum nonzero entry,  $\DD^*$ is the optimal balancing matrix,  $\kappa(\cdot)$ has the same meaning as above, and $\eps$ is the measure of the balancing error, as formally defined in Definition~\ref{def:eps_balancing}.

Notably, our running times depend logarithmically on both the target accuracy and the magnitude of the entries in the optimal balancing or scaling. This implies that if the optimal solution has quasi-polynomially bounded entries, our algorithms run in nearly linear time $\Otil(m \ln (1/\eps))$ (ignoring logarithmic factors involving the entries of the input matrix). This includes, for instance, the case when input matrix has all its entries positive or, in case of matrix balancing, if there just exists a single row/column pair with all positive entries.

However, there are matrices for which $\kappa$ can be exponentially large (in $n$). For the case of such matrices we develop algorithms with negligible dependence on $\kappa$. These algorithms are based on interior point methods, with appropriately chosen barriers, commonly used in exponential programming~\cite{ben2001lectures}. We show that the linear system solves required by the interior point method every iteration can be reduced via Schur complementing to approximately solving a Laplacian system, which can be done in nearly linear time using any standard Laplacian solver~\cite{SpielmanT04,KoutisMP10,KoutisMP11,KelnerOSZ13,CohenKMPPRX14,KyngLPSS16,KyngS16}. 
This yields a running time of 
$$\Otil\left( m^{3/2} \ln \frac{\wA}{\eps} \right),$$
 where $\wA$ is the ratio between the largest and smallest nonzero entry of $\AA$. 

%% file: approach.tex
\subsection{Our Approach}

We approach the scaling and balancing problems by developing a continuous optimization based perspective on them. More precisely, we solve both matrix scaling and balancing problems by casting them as tasks of minimizing certain corresponding convex functions. In fact, in the case of the balancing problem, that function is directly inspired by the one used in~\cite{KKS97}; for the scaling problem, it is function derived from the one used in~\cite{KK96}. 

Since our goal is to obtain logarithmic -- instead of polynomial -- dependence on the (inverse of the) desired accuracy $\eps$, it would be tempting to use well-known tools for convex programing, such as ellipsoid method or interior point method. However, these methods are, a priori, computationally expensive. This motivates us to look for different, more direct approaches.

To this end, we develop a technique for minimizing a broader class of functions that we call
\textit{second-order robust (with respect to $\ell_\infty$)}. Intuitively, this class corresponds to
functions whose Hessians do not change too much within any unit $\ell_\infty$-ball.
And the consequence of that property that will be crucial for us is that local quadratic approximation of such functions at
any given point is relatively accurate within the unit $\ell_\infty$ neighborhood of that
point. As a result, iteratively optimizing the local approximation around
the current point, while staying within that $\ell_\infty$ neighborhood, will be guaranteed to
make progress towards minimizing the function. This iterative procedure can be viewed as a ``box-constrained'' variant of the Newton's method. 

A priori, performing a single step of such a box-constrained Newton's method, i.e., minimizing a quadratic function subject to box constraints might be a computationally costly task. We show, however, that it suffices to implement a weaker primitive, which we call a \textit{$k$-oracle}. That primitive corresponds to (approximately) minimizing a quadratic function within a region that is within a factor of $k$ larger than the target $\ell_\infty$-ball. Once such a $k$-oracle is implemented efficiently, we can compute the global optimum of our second-order robust function using a small number of calls to it. More precisely, we show that one can minimize a convex function $f$ that is second-order robust with respect to $\ell_\infty$  to within $\epsilon$ additive error from optimum in 
\begin{equation}
\label{eq:k_oracle_iterations}
O\left((k R_\infty+1)  \log \left(\frac{f(x_0) -
f(x^*)}{\epsilon}\right)\right)
\end{equation} iterations, where each iteration consists of one call to the $k$-oracle, $x_0$ is the starting point, $x^*$ is the minimizer of $f$, and $R_\infty$ is the $\ell_\infty$ radius of the level set of $x_0$.

In the light of the above, the main technical difficulty remaining is obtaining an efficient implementation of a
$k$-oracle. We show that for functions whose Hessian is symmetrically diagonally
dominant, with nonzero off-diagonal entries, or {\SDDM} for short\footnote{Such matrices 
can essentially be viewed as a Laplacian matrix plus a nonnegative diagonal.}, we can implement
a $k$-oracle, with $k=O(\log n)$, in time that is nearly linear in the sparsity of the Hessian.
We build here on the strategy underlying the Laplacian solver of Lee, Peng and Spielman~\cite{LeePS15}.  Specifically, we carefully lift the solutions corresponding to coarser (and smaller) approximations of the underlying matrix to the desired solutions corresponding to the initial matrix in a way that does not allow these lifted solutions to exceed the boundaries of a $O(\log n)$-radius $\ell_\infty$-ball.

Once the above optimization framework is developed,  applying it to the scaling and balancing problems is
fairly straightforward. It boils down to verifying that the functions that capture the respective problems are indeed
second-order robust and have an SDD Hessian, and then bounding all the relevant quantities that \eqref{eq:k_oracle_iterations} involves.

\paragraph{Independent Work} Finally, we note that Allen-Zhu, Li, Oliveira, and Wigderson ~\cite{allen2017much} obtained independently very similar results for the exact version of the problem. The running time of the algorithms they develop have a bit worse dependence on $m$, but they were able to establish better absolute bounds on $\kappa$ (in terms of the problem parameters and the magnitude of the input entries) for the general, non-doubly stochastic variant of the matrix scaling problem.

\subsection{Roadmap}
The rest of the paper is organized as follows.
First, we introduce relevant notation and concepts in Section~\ref{sec:prelim}.
Then, in Section~\ref{sec:mainloop} we formally introduce the class of convex functions we call
\textit{second-order robust with respect to $\ell_\infty$}. For these, we develop a specific optimization primitive called \textit{box-constrained Newton method}.

We describe how we can apply the primitive from Section~\ref{sec:mainloop} to matrix
balancing and scaling in Section~\ref{sec:scb}, by reducing these problem to a convex
function minimization with favorable structure.
In order to complete our algorithm, in Section~\ref{sec:nearlylineartime}, we show how to efficiently implement an iteration of the box-constrained Newton in the special case where the Hessian of the function is {\SDDM}.
In Section~\ref{sec:ipm} we provide a different approach for balancing and scaling based on
interior point methods.
Supplementary proofs and technical details are presented in the Appendix.

%% file: prelim.tex

\section{Preliminaries}\label{sec:prelim}

\subsection{Notations}

\paragraph{Vectors} We let $\allzeros, \allones \in \R^n$ denote the all zeros and all ones
vectors, respectively. When it is clear from the context, we apply scalar operations to
vectors with the interpretation that they are applied coordinate-wise. 

\paragraph{Matrices} We write matrices in bold. We use $\mI$ to denote the identity matrix, and $\mzero$ to denote the zero matrix. Given a matrix $\AA$, we denote its number of nonzero entries by $\textnormal{nnz}(\AA)$. When it is clear from the context, we use $m$ to denote the the number of nonzeros; similarly, we use $n$ to denote the dimension of the ambient space.

 We denote by $\uA$ the sum of entries of $\AA$, by $\lA$ the minimum nonzero
 entry of $\AA$, and by $\wA$ the ratio between these quantities. We use
 $\supp(\AA)$ to denote the set of pairs of indices $(i,j)$ corresponding to the
 nonzero entries of $\AA$. Given a matrix $\AA$, we define $r_{\AA}  = \AA
 \allones$ to be the vector consisting of row sums, and $c_{\AA} = \AA^\top
 \allones$ to be the vector consisting of column sums. For a positive diagonal
 matrix $\AA$ we denote the maximum ratio between its diagonal elements by
 $\kappa(\AA)$.

\paragraph{Positive Semidefinite Ordering and Approximation} For symmetric matrices $\AA, \BB \in \R^{n\times n}$ we use $\AA \preceq \BB$ to represent the fact that that $x^\top \AA x \leq  x^\top \BB x$, for all $x$. A symmetric matrix $\AA \in \R^{n\times n}$ is positive semidefinite (PSD) if $\AA \succeq 0$. We use $\preceq, \succ, \prec$ in a similar fashion.
For vectors $x$, we define the norm $\Vert x \Vert_{\AA} = \sqrt{x^\top \AA x}$.
Given two PSD matrices $\AA$ and $\BB$, and a parameter $\alpha > 0$,  we use $\AA
\approx_\alpha \BB$ to denote the fact that $ e^{-\alpha}\cdot \BB \preceq \AA \preceq
e^{\alpha} \cdot \BB$.

\paragraph{Laplacian and SDD matrices} A family of matrices that will play an important role in this paper are symmetric diagonally dominant (SDD) matrices. These are
matrices $\AA$, that symmetric and, moreover, have each diagonal entry be larger than the sum of
absolute values of the corresponding row entries. That is, for every $i$
$$A_{ii}\geq \sum_{j\neq i} |A_{ij}|.$$
A special case of SDD matrices are Laplacian matrices, which have negative off-diagonal
entries and the sum of each row is required to be zero. The crucial fact about these matrices is that one can exploit their structure to solve linear systems in them in time that is only nearly linear \cite{SpielmanT04,KoutisMP10,KoutisMP11,KelnerOSZ13,CohenKMPPRX14,KyngLPSS16,KyngS16}.

\paragraph{Diagonal Matrices} For $x \in \R^n$ we denote by $\mdiag(x) \in \R^{n\times n}$ the
diagonal matrix where $\mdiag(x)_{ii} = x_i$. Given a nonnegative diagonal matrix $\DD$, we
use $\kappa(\DD)$ to denote the ratio between its largest and smallest entry. We will
overload notation and, for any matrix $\AA\in\R^{n\times n}$, use $\mdiag(\AA)$ to denote the
main diagonal of $\AA$, that is $(\mdiag(\AA))_{ii} = \AA_{ii}$ and $(\mdiag(\AA))_{ij}=0$
for $i\neq j$.

\paragraph{Gradients and Hessians} Given a function $f$ we denote by $\nabla
f(x)$ its gradient at $x$, and by $\nabla^2 f(x)$ its Hessian at $x$. When the
function is clear from the context, we also use $\HH_x$ to denote its Hessian at $x$.

\paragraph{Block Matrices} As part of our algorithms, we will consider partitioning the
coordinates of vectors into sets of indices $F$ and $C$. When we compute the quadratic form
of a matrix with these vectors, we need to be able to reason about how values in each
component interact with the rest of the vector. For that reason it is convenient to denote
the block form notation for a matrix $\AA$ as:
$$\AA = \begin{bmatrix}
    \AA_{[F,F]} & \AA_{[F,C]} \\
    \AA_{[C,F]} & \AA_{[C,C]} \end{bmatrix}\,{.}$$

\paragraph{Schur Complements}
For a matrix $\AA\in\R^{n\times n}$ and a partition of its
indices $(F,C)$, the Schur complement of F in $\AA$ is defined as
$$\Sc(\AA,F) \defeq \AA_{[C,C]} - \AA_{[C,F]}\AA_{[F,F]}^{-1}\AA_{[F,C]}\,{.}$$
The exact use of Schur complements will become clear in
Sections~\ref{sec:nearlylineartime},\ref{sec:ipm}. These are objects that naturally arise
during Gaussian elimination for the solution of linear systems. By pivoting out variables $F$
the remaining system to solve for variables of $C$ is exactly the Schur complement of $F$ in
$\AA$.

%% file: newton.tex
\section{Box-Constrained Newton Method for Second-Order Robust Functions}\label{sec:mainloop}

The central element of our approach is developing an efficient second-order
method based minimization framework for a broad class of functions that we will call second-order robust with respect to $\ell_\infty$. To
motivate the choice of this class, recall that second-order methods for function
minimization are iterative in nature, and they boil down to repeated minimizing
the local quadratic approximation of the function around the current point.
Consequently, in order to obtain meaningful guarantees about the progress made
by such methods, one needs to ensure that this local quadratic approximation constitutes a good approximation of the function not only at the current point but also in a reasonably large neighborhood of that point. The most natural way to obtain such a guarantee is to ensure that the Hessian of the function (which is the basis of our local quadratic approximations) does not change by more than a constant factor in that neighborhood. As a result, the
functions we are interested in optimizing in this paper are the ones that satisfy that property in an
$\ell_\infty$-ball around the current point. This is formalized in the following definition.

\begin{defn}[Second-Order Robust w.r.t. $\ell_\infty$]\label{def:sor}
 We say that a convex function $f : \R^n \rightarrow \R$ is {\em second-order robust (SOR) with respect to $\ell_\infty$} if, for any $x, y \in \R^n$ such that $\Vert x-y \Vert_\infty \leq 1$,
$$\nabla^2 f(x) \approx_2 \nabla^2 f(y),\text{ that is, }
\frac{1}{e^2}\nabla^2 f(x)\pleq \nabla^2f(y)\pleq e^2\nabla^2f(x)\,{.} $$
\end{defn}

Note that the size of the $\ell_\infty$-ball, as well as the exact factor
by which the Hessian is allowed to change, are chosen somewhat arbitrarily -- all choices of the constants can be made equivalent via an appropriate rescaling. Moreover, even if these quantities are not constant,
they would only appear in the running time as a small polynomial factor.

Now, the above definition suggests a natural framework for optimizing such functions. Namely, in every iteration, we optimize a local quadratic approximation of the function within a unit $\ell_\infty$-ball around the current point. As we will see shortly, this approach can be rigorously analyzed. In particular, our key technical result is that if we apply this approach to an SOR function whose Hessians has additionally a special structure, i.e., those for which the Hessian is, essentially, a symmetric diagonally dominant (SDD) matrix, we can implement every iteration in time nearly linear in the number of nonzero entries of the Hessian. This leads to running time bounds captured by  the following theorem.

\begin{theorem}[Minimizing Second-Order Robust Functions w.r.t $\ell_\infty$]
\label{thm:sor_sdd}
Let $f : \R^n \rightarrow \R$ be a second-order robust (SOR) function with respect to
$\ell_\infty$, such that its Hessian is symmetric diagonally dominant (SDD) with nonpositive
off-diagonals, and has $m$ nonzero entries.  Given a starting point $x_0\in\R^n$ we can
compute a point $x$, such that $f(x) - f(x^*)\leq \eps$, in time
$$\Otil\left( (m+T) R_\infty \log \left(\frac{f(x_0) - f(x^*)}{ \epsilon}
\right)\right)\,{,}$$
where $x^*$ is a minimizer of $f$, $R_\infty = \sup_{x : f(x) \leq f(x_0)} \Vert x - x^* \Vert_\infty $ is the
$\ell_\infty$ diameter of the corresponding level set of $f$, and $T$ is the
time required to compute the Hessian.
\end{theorem}

Note that the bounds provided by the above theorem are, in a sense, the best possible for any kind of approach that relies on repeated minimization of a local approximation of a function in an $\ell_\infty$-ball neighborhood. In particular, as each step can make a progress of at most $1$ in $\ell_\infty$-norm towards the optimal solution, one would expect the total number of steps to be $\Omega(R_\infty)$. 

It turns out that the above theorem is all we need to establish our results for scaling and balancing problems (except the ones relying on the interior point method). That is, these results can be obtained by direct application of the above theorem to an appropriate SOR function. We provide all the details in Section \ref{sec:scb}.

Now, the first step to proving the above Theorem \ref{thm:sor_sdd} is to view each iteration of our iterative minimization procedure as a call to a certain oracle problem. 

\begin{definition}\label{def:oracle}
We say that a procedure $\oracle$ is a $k$-oracle for a class of matrices $\mathcal{M}$, if
on input $(\AA, b)$, where $\AA \in \mathcal{M} \subseteq \R^{n \times n}$, and $b \in \R^n$,
returns a vector $\tilde{z}$ satisfying 
\begin{enumerate}[label=\normalfont (\arabic*)]
\item $\Vert\tilde{z} \Vert_\infty \leq k$, and
\item $ \frac{1}{2} \tilde{z}^\top \AA \tilde{z} + b^\top \tilde{z} \leq \frac{1}{2} \cdot
\min_{\Vert z\Vert_\infty \leq 1}\left(\frac{1}{2}z^\top \AA z + b^\top
z\right)$\,{.} 
\end{enumerate}
\end{definition}

Note that the minimum of the left-hand side of Condition (2) above is always
non-positive. This is desired, since this expression is supposed to measure our function minimization progress.

Observe that minimizing the function $ \frac{1}{2} z^\top \AA z + b^\top z$
without any constraints on $z$ corresponds to solving a linear system
$\AA z=-b$. So, one can view the $k$-oracle problem as a certain generalization of linear system solving. Specifically, it is a task in which we aim to find a point in the $\ell_\infty$-ball of diameter $k$ around the origin that is closest (in a certain sense) to the solution to that linear system. If $b$ is sufficiently small, the $k$-oracle problem corresponds directly to solving that system. 

One can view the parameter $k$ as the measure of the ``quality" of our $k$-oracle. The
smaller it is, the faster convergence the overall procedure will have. Importantly, however, the value of $k$ impacts only the convergence and not the quality of the final solution. The following theorem makes this relationship precise.
\begin{theorem}\label{thm:newtonoracle}
Let $f:\R^n\rightarrow \R$ be a  function that is second-order robust with
respect to $\ell_\infty$. Let $\mathcal{O}$ be a $k$-oracle for $\{\nabla^2 f(x)
: x \in \R^n\}$, along with an initial point $x_0 \in \R^n$ and an accuracy parameter $\epsilon$. 
Let $R_\infty = \sup_{x : f(x) \leq f(x_0)} \Vert x - x^* \Vert_\infty$, where $x^*$ is a
minimizer of $f$. Then one can produce a solution $x_T$ satisfying $f(x_T)-f(x^*) \leq \epsilon$ using
 $$O\left((k R_\infty+1) \log \left(\frac{f(x_0) - f(x^*)}{\epsilon}\right) \right)$$  calls to $\mathcal{O}$.\end{theorem}

We present the proof of this theorem in Section~\ref{sec:pfofnewtonoracle} of
the Appendix.
In Section~\ref{sec:nearlylineartime} we design an efficient $k$-oracle, with $k=O(\log n)$, for
the family of SDD matrices.
Combining Theorem~\ref{thm:subproblem} with Theorem~\ref{thm:newtonoracle}
immediately gives the proof of Theorem~\ref{thm:sor_sdd}. We remark that while
Theorems~\ref{thm:sor_sdd},~\ref{thm:newtonoracle} are stated and proved for
functions defined over $\R^n$, they can be extended in a straightforward way to
hold when $f$ is defined over an arbitrary closed, convex set.

%% file: bal.tex

\section{Matrix Scaling and Balancing}\label{sec:scb}

Having developed our main optimization primitives, we can develop efficient
algorithms for matrix scaling and matrix balancing.
Our approach is essentially the same for both problems, and differs only in technical details.

At the high level, we will construct convex functions with optima corresponding to exact scaling/balancing of
the input matrix. 
Moreover, the gradient of these functions at a specific scaling/balancing of
the matrix will be directly related to the quality of this particular
scaling/balancing. This will allow us to prove that approximately
optimal points correspond to $\eps$-approximate scaling/balancing. The fact that that these functions
are second-order robust with respect to $\ell_\infty$ makes it sufficient to apply the
optimization method from Section~\ref{sec:mainloop}. To complete the algorithm and
its running time analysis, we need then to address two issues.

Firstly, proving running time bounds for this method requires an upper bound on the $\ell_\infty$ radius
of the level set of the initial point, i.e. the $R_\infty$ parameter defined in Theorem~\ref{thm:newtonoracle}.
Depending on the structure of the
matrix, there are several different bounds that one can prove, depending only on parameters
of the original problem. However, the most interesting case is when we are promised
that the exact scaling/balancing of the matrix is ``small'' (in the sense that the
\emph{ratio} between factors is, say, polynomial). In that case, we can regularize the
function to turn this promise into a guarantee for the size of the level set without
sacrificing too much accuracy. Moreover, by using a simple doubling approach, we can make the algorithm 
not require explicit knowledge of the value of this parameter, and it will only appear as a factor in the final running time
of the algorithm.

Secondly, we need to ensure that we can efficiently implement $k$-oracles for the
Hessians of these functions. In our case, this boils down to proving that these Hessians are SDD
matrices with sparsity equal to that of the input matrix, and then build on the existing Laplacian solving work.
For the remainder of this section, we define the convex functions that we need optimize, show how to regularize them, and prove bounds on the corresponding $R_\infty$ parameters. We will describe and analyze the implementation of a $O(\log n)$-oracle in Section~\ref{sec:nearlylineartime}.

\subsection{Matrix Scaling}

We now formally define the scaling problem, along with the notion of $\eps$-scaling.
\begin{defn}[Matrix Scaling]
Let $\AA\in\R^{n\times n}$ be a nonnegative matrix and $r, c\in\R^{n}$ be
vectors such that $\sum_{i=1}^{n} r_i = \sum_{j=1}^{n} c_i$, and $\Vert r \Vert_\infty, \Vert c \Vert_\infty \leq 1$\footnote{In literature we also encounter this problem for non-square matrices; however solving squares is sufficient, since given $\AA\in\R^{n \times c}$, we can reduce to this instance by scaling the square matrix $\begin{bmatrix} \mzero_{c,c} & \AA^\top \\ \AA & \mzero_{r,r} \end{bmatrix}$. The upper bound on $r$ and $c$ is harmless, since for larger values we can always shrink all of $\AA$, $r$, $c$ and $\epsilon$ by the same factor in order to enforce this constraint.}. 
We say that
two nonnegative diagonal matrices $\XX$ and $\YY$ $(r,c)$-scale $\AA$ if the matrix $\MM =
\XX \AA \YY$ satisfies $\MM \allones = r$ and $\MM^\top \allones = c$, i.e. row
$i$ sums to
$r_i$ and column $j$ sums to $c_j$ for every $i,j$.
\end{defn}

\begin{defn}[$\epsilon$-$(r,c)$ scaling]
\label{def:eps_scaling}
Given nonnegative $\AA$ and positive diagonal matrices $\XX, \YY$, we say that $(\XX, \YY)$ is an $\epsilon$-$(r,c)$ scaling (or $\epsilon$-scaling, when $r$ and $c$ are clear from the context) for matrix $\AA$ if the matrix $\MM = \XX \AA \YY$ satisfies
$$  \Vert r_{\MM} - r \Vert_2^2 + \Vert c_{\MM}  - c \Vert_2^2 \leq \epsilon\,{.}$$
\end{defn}

\begin{defn}[Scalable and Almost-Scalable Matrices] A nonnegative matrix $\AA$, is called
$(r,c)$-scalable, if there exist $\XX$ and $\YY$ that $(r,c)$-scale $\AA$. It is called
almost $(r,c)$-scalable if for every $\eps >0$, there exist $\XX_\eps$ and $\YY_\eps$ that
$\eps$-$(r,c)$ scale $\AA$.
\end{defn}

There are well-known necessary and sufficient conditions about the scalability of $\AA$
stated in the following lemma.
\begin{lemma}[\cite{LSW98}]\label{lem:scalcond}
A nonnegative matrix $\AA$ is exactly $(r,c)$-scalable iff for every zero
minor $Z\times L$ of $\AA$,
\begin{enumerate}[label=\normalfont (\arabic*)]
\item $\sum_{i\in Z^c}r_i \geq \sum_{j\in L} c_j$.
\item Equality in (1) holds iff $Z^c\times L^c$ is also a zero minor.
\end{enumerate}
A nonnegative matrix $\AA$ is almost $(r,c)$-scalable iff Condition (1) above holds.
\end{lemma}

We will cast matrix scaling as a convex optimization problem and show that applying the
method from section~\ref{sec:mainloop} yields a good approximate scaling.

\begin{theorem}\label{thm:scaling}
Let $\AA$ be a matrix, that has an $(r,c)$ scaling $(\UU^*, \VV^*)$. Then, we can compute an $\epsilon$-$(r,c)$ scaling of $\AA$ in time
$$\Otil\left(  m \ln(\kappa(\UU^*) + \kappa(\VV^*))  \ln^2 
      (\uA/\epsilon)\right)\,{.}$$
\end{theorem}
This implies that if $\UU^*$ and $\VV^*$ are, say, quasi-polynomially bounded, we can find an
approximate scaling in nearly linear time. 
If fact, we can generalize this statement to obtain a similar result for the case of
approximate scalings. This is made precise in Theorem~\ref{thm:escaling}.

\subsubsection{Matrix Scaling via Convex Optimization}
Recall that we want to encode the matrix scaling problem as a an instance of minimizing of a certain convex function. Given the input matrix $\AA$, the function we want to consider is: 
\begin{equation}\label{eq:scaling_fn}
f(x,y) = \sum_{1 \leq i,j \leq n} \AA_{ij} e^{x_i - y_j} - \left(\sum_{1 \leq i \leq n}
r_ix_i - \sum_{1 \leq j \leq n} c_jy_j\right)\,{.}
\end{equation}
We want to argue now that computing an (approximate) scaling of the matrix $\AA$ can indeed be recovered from an (approximate) minimum of the above function. Specifically, we want to establish the following theorem. 
\begin{theorem}\label{thm:escaling}
Suppose that there exist a point $z^*_\epsilon = (x^*_\epsilon, y^*_\epsilon)$ for which
$f(z^*_\epsilon)-f^* \leq \epsilon^2/(3n)$ and $\|z^*_\epsilon \|_\infty \leq B$. Then we can
compute an $\epsilon$-$(r,c)$ scaling of $\AA$ in time
$$\Otil\left(  m B  \ln^2 
      (\uA/\epsilon)\right)\,{.}$$
\end{theorem}
The proof is straightforward given the lemmas below and is presented in Section~\ref{sec:mainscal} of the Appendix.
First, we will prove that approximate optimality of $f$ implies an approximate scaling of
the matrix.
\begin{lemma}\label{lem:scaling_accuracy}
Let $\AA$ be an $\epsilon$-scalable matrix. Let $f^* = \inf_{(x,y)} f(x,y)$. Then, a pair of vectors $(x,y)$ satisfying $f(x,y) - f^* \leq \epsilon^2 / 3n$, for $0 < \epsilon \leq 1$, yields an $\epsilon$-$(r,c)$ scaling of $\AA$:
$$\MM = \mdiag(\exp(x))\cdot \AA \cdot \mdiag(\exp(y))\,{.}$$
\end{lemma}
Note that we compare the value of $f(x,y)$ to its infimum, as for the case of almost scalable
matrices it is possible that this value is attained only in the limit. 

To prove the above lemma, we first look at the first and second order derivatives of $f$.
\begin{lemma}\label{lem:gradhessscal}
Let $\MM$ be the matrix obtained by scaling $\AA$ with vectors $(x,y)$, i.e. $\MM = \mdiag(\exp(x)) \cdot \AA \cdot \mdiag(\exp(y))$. The gradient and Hessian of $f$ satisfy the identities:
\begin{align*}
\nabla f(x,y) &= 
\begin{bmatrix*}[r] r_\MM\\ -c_\MM 
\end{bmatrix*} -
\begin{bmatrix*}[r]r \\
-c
\end{bmatrix*},\\
\nabla^2 f(x,y) &= 
\begin{bmatrix*}[r] \mdiag(r_\MM) & -\MM\\ -\MM^\top & \mdiag(c_\MM) 
\end{bmatrix*}.
\end{align*}
\end{lemma}

We can observe that any $(x,y)$ for which $\nabla f(x,y)$ is equal to $0$ yields diagonal matrices that exactly scale $\AA$. Moreover, this statement also holds in an approximate sense. One can prove that a large gradient in $\ell_2$ norm implies that the current point is far from optimal in function value. Making this statement precise, allows us to prove Lemma~\ref{lem:scaling_accuracy}. The technical details are presented in Section~\ref{sec:pfofscaling_accuracy} of the appendix.

\subsubsection{Regularization for Solving via Box-Constrained Newton Method}

It is straightforwards to verify that the function we are minimizing (defined in
Equation~\ref{eq:scaling_fn}), satisfies the requirements necessary for us to be able to apply the tools
from Section~\ref{sec:mainloop}.

\begin{lemma}\label{lem:fscallem}
The function $f$ defined in \eqref{eq:scaling_fn} is convex, second-order robust with respect to $\ell_\infty$,
and its Hessian is \SDDM.
\end{lemma}
\begin{proof}
The Hessian of the function $f$ (cf. Lemma \ref{lem:gradhessscal}) is clearly a Laplacian matrix. 
Therefore, it is positive semi-definite and thus $f$ is convex.  To prove that it is second-order
robust, we notice that adding some $z$ with $\|z\|_\infty\leq 1$ to the current scaling
corresponds to multiplying each row and column by some factor between $1/e$ and $e$.  By
writing down the quadratic form of the Hessian, $v^\top \nabla^2 f(x) v = \sum_{i,j} \MM_{ij}
(v_i - v_j)^2$, we observe that each $\MM_{ij}$ will only be multiplied by some factor
between $1/e^2$ and $e^2$, proving that $$\frac{1}{e^2}\nabla^2 f(x) \pleq \nabla f(x + z)
\pleq e^2 \nabla^2 f(x),$$ concluding the proof.
\end{proof}

One should observe, however, that Theorem~\ref{thm:newtonoracle} requires bounding the radius of the entire level set containing our initial point and not merely the distance to some (approximate) minimizer of our function $f$. This means that the existence of an (approximate) minimizer that is close to our initial point is not sufficient to apply Theorem \ref{thm:newtonoracle}. To circumvent that problem, we regularize the function $f$ by adding to it a term that, on one hand,  has a relatively small impact on the additive error we can achieve, but, on the other hand, ensures that the entire relevant level set is contained in some sufficiently small $\ell_\infty$-ball around our initial point. The following lemma makes these statements precise. Its proof appears in Section~\ref{sec:pfofscalreg} of the Appendix.

\begin{lemma}\label{lem:scalreg}
Let $z^*_\epsilon = (x^*_\epsilon, y^*_\epsilon)$ be a point for which
$f(z^*_\epsilon)-f^* \leq \epsilon^2/(3n)$ and $\|z^*_\epsilon \|_\infty \leq B$. Then, the regularization of $f$ defined as
\begin{equation}\label{eq:reg_scal}
\widetilde f(x,y) = f(x,y) + \frac{\eps^2}{36n^2 e^B}\left(\sum_i(e^{x_i} +
e^{-x_i})+\sum_j(e^{y_j} - e^{-y_j})\right)
\end{equation}
satisfies the following properties
\begin{enumerate}[label=\normalfont (\arabic*)]
\item $\widetilde{f}$ is second-order robust with respect to $\ell_\infty$ and its Hessian is SDD,
\item $f(z) \leq \widetilde f(z)$, and there is a point $\widetilde z^*$ such that $\widetilde{f}(\widetilde z^*) \leq f^* + \frac{\eps^2}{9n}$,
\item for all $z'$ such that $\widetilde{f}(z') \leq \widetilde{f}(0)$,
$\|z'\|_\infty = O(B \log(n \uA/\epsilon))$.
\end{enumerate}
\end{lemma}

Theorems \ref{thm:scaling} and \ref{thm:escaling} follow from applying Theorem~\ref{thm:sor_sdd} to the  regularized function defined in \eqref{eq:reg_scal}, and then combining it with the guarantees of Lemmas~\ref{lem:scaling_accuracy} and~\ref{lem:scalreg}.
The complete proof is presented in Section~\ref{sec:mainscal} of the Appendix.
We note that we don't need an explicit knowledge of an a priori bound on $B$. We can simply run our algorithm
repeatedly, doubling our guess at the value of $B$ each time.  This will not increasing the overall running time by more
than a factor of two.

\subsubsection{Bounding the Magnitude of the Optimal and Approximately Optimal Scalings for
Doubly Stochastic Scaling}

In order to provide bounds for the magnitude of the scaling factors that only depend on the
parameters of the initial problem, we refer to
the following lemmas from \cite{KK96} for the case of double stochastic (i.e. (1,1)) scaling.
\begin{lemma}[Lemma 1 of \cite{KK96}]
If $\AA$ is strictly positive, then it can be scaled to doubly stochastic by
diagonal matrices $\UU$, $\VV$ with  $\log(\kappa(\UU) + \kappa(\VV))\leq
O(\log(\wA))$.
\end{lemma}
\begin{lemma}[Corollary 1 of \cite{KK96}]\label{lem:kkn}
If $\AA$ is scalable, then it can be scaled to doubly stochastic by 
diagonal matrices $\UU$, $\VV$ with  $\log(\kappa(\UU) + \kappa(\VV))\leq
O(n\log(\wA))$.
\end{lemma}
For almost scalable matrices, there can be arbitrarily good solutions,
using arbitrarily large scaling factors. To prove bounds on the runtime of finding an
approximate doubly-stochastic matrix, we will have to explicitly demonstrate an vector that
approximately minimizes function $f$ while having small $\ell_\infty$ norm.
\begin{lemma}\label{lem:almostdss}
If $\AA$ is almost-doubly-stochastic scalable, then there exist points $(x,y)$ such that
$f(x,y) - f^* \leq \eps^2/3n$, such that $\|(x,y)\|_\infty \leq O(n\log(n\wA/\eps))$.
\end{lemma}
The proof of the lemma is presented in Section~\ref{sec:pfofalmostdss}.

For the general case of $(r,c)$-scaling we refer to the recent lemmas from the
parallel work of~\cite{allen2017much}. The assumption that the scaling targets
are integral is mild, since one can approximate real numbers by rational ones
which can then by scaled to be integral (the dependence on this scaling is
logarithmic).

\begin{lemma}[Lemma 3.3 of \cite{allen2017much}]\label{lem:rcbounds}
If $\AA$ is almost $(r,c)$-scalable with $r$, $c$ being integral, then it can be
$\eps$-scaled by diagonal matrices $\UU$, $\VV$ with  $\log(\kappa(\UU) +
\kappa(\VV))\leq O(n\log(n\wA\|r\|_1/\eps))$.
\end{lemma}

\subsection{Matrix Balancing}\label{sec:mbal}
Our approach for the balancing problem is completely analogous to the one we used for the scaling problem. There are only
minor technical differences. To state them, we first formally define the problem and the notion of approximation we are considering for it.
\begin{defn}[Matrix Balancing]
Let $\AA$ be a square nonnegative matrix.
We say that $\AA$ is balanced if the sum of each row is equal to the sum of the corresponding
column, i.e. $r_\AA = c_\AA$.
We say that a nonnegative diagonal matrix $\DD$ balances $\AA$ if the matrix $\MM = \DD \AA
\DD^{-1}$ is balanced.
\end{defn}

\begin{defn}[$\epsilon$-Balanced Matrix \cite{KKS97}]
\label{def:eps_balancing}
We say that a nonnegative matrix $\MM \in \R^{n\times n}$ is $\epsilon$-balanced if
$$ \frac{\Vert r_{\MM} - c_{\MM} \Vert_2}{ \sum_{1\leq i, j\leq n} \MM_{ij}} =   \frac{\sqrt{\sum_{i=1}^n ((r_{\MM})_i - (c_{\MM})_i)^2}}{\sum_{1\leq i, j\leq n} \MM_{ij}} \leq \eps.$$
\end{defn}
Observe that this definition is invariant to a global scaling of all the entries of the matrix
by some factor.
There is a very simple condition that characterizes the set of matrices that can be balanced
\begin{lemma}[\cite{KKS97}]
A nonnegative matrix $\AA\in\R^{n\times n}$ can be balanced if and only if the graph with
adjacency matrix $\AA$ is strongly connected.
\end{lemma}
In the case when the graph is not strongly connected, the matrix can have its rows and
columns rearranged so as to be written as a lower triangular block matrix with strongly connected
diagonal blocks. The reason no exact balancing exists is that off diagonal block elements will always
create imbalances.
This, however, is not an obstacle for approximately balancing the matrix.
Once we balance the diagonal blocks, we can set all of the off-diagonal block entries to a very small value, say $\eps/n$, so that they don't cause significant imbalances. This corresponds to implicitly scaling the block rows and collumns by a very large amount, making the off-diagonal entries arbitrarily close to zero.
Therefore, since the case of matrices that cannot be exactly balanced is easy to detect, and can be easily reduced to the exactly balanceable case, from now on we consider only matrices that can be balanced, and therefore represent strongly connected graphs.

We can now state our main theorem for this section, which follows our initial discussion.
\begin{theorem}\label{thm:balancing}
Let $\AA$ be a matrix that can be balanced by the diagonal matrix $\DD^*$. Then, we can compute an $\epsilon$-approximate balancing of $\AA$ in time
$$\Otil(m \log \kappa(\DD^*) \log^2(\wA/\epsilon) )\,{.}$$
\end{theorem}
This immediately implies that if $\DD^*$ is, say, quasi-polynomially conditioned, we can find an approximate balancing in nearly linear time. 

Again, we can generalize this result to hold for approximate balancings. We make this
statement precise in Theorem~\ref{thm:ebalancing}.

\subsubsection{Reducing Matrix Balancing to Convex Optimization}
Similarly to the case of the scaling problem, we encode this problem as a minimization of an appropriately constructed convex function. The function we consider here is 
\begin{equation}\label{eq:fbal}
f(x) = \sum_{1\leq i,j\leq n} \AA_{ij} e^{x_i - x_j},
\end{equation}
and this function was already defined in~\cite{KKS97}. Similarly to the case of matrix scaling, we will show that (approximately) minimizing this function corresponds to (approximately) balancing the
matrix $\AA$.  For the rest of this section, we will define $f_*$ to be the infimum value of
$f$ in its domain, that is $f_* = \inf_{x} f(x)$. The main theorem of this section is the
following.

\begin{theorem}\label{thm:ebalancing}
Suppose that there exists a point $x$ such that $f(x) \leq f_* + \eps^2\lA/24$, and $\| x
\|_{\infty} \leq B$. Then, we can compute an $\eps$-approximate balancing of $\AA$ in time
$$\Otil(m  B \log^2(\wA/\epsilon) )\,{.}$$
\end{theorem}

Similarly to the matrix scaling case, the proof of this theorem follows directly from the key lemmas presented below. The proof is presented in Section~\ref{sec:mainbal} of the Appendix. First, we prove that small additive
error in function optimization implies an approximate balancing for $\AA$.
\begin{lemma}\label{lem:balancing_requirement}
Consider a matrix $\AA$ and the corresponding function $f$. Any vector $x$ satisfying $f(x) -
f_* \leq \epsilon^2 \lA / 8$ yields an $\epsilon$-approximate balancing of $\AA$:
$$\MM = \mdiag(\exp(x)) \cdot \AA  \cdot \mdiag(\exp(-x))\,{.}$$
\end{lemma}

Proving the lemma requires computing the first and second order derivatives of $f$.
\begin{lemma}\label{lem:gradhess} Let $\MM$ be the matrix obtained by balancing $\AA$ with
the vector $x$, which corresponds to $\MM = \mdiag(\exp(x)) \cdot \AA \cdot \mdiag(\exp(-x))$. The gradient and Hessian of $f$ satisfy the identities:
\begin{align*}
\nabla f(x) &= r_{\MM}-c_{\MM} \,{,} \\
\nabla^2 f(x) &= \mdiag(r_{\MM}+c_{\MM}) - (\MM+\MM^\top)\,{.}
\end{align*}
\end{lemma}
Intuitively, since the gradient is $0$ precisely when the corresponding point produces an
exact balancing, a small gradient should imply a good approximate balancing. This guides the
proof of Lemma~\ref{lem:balancing_requirement}. We will prove that a large gradient
corresponds to being able to significantly decrease the function value, thus contradicting
the approximate optimality of the point.
The technical details are presented in Section~\ref{sec:pfofbalreq}.

 \subsubsection{Regularization for Solving via Box-Constrained Newton Method}
We observe that the function $f$ defined in~\eqref{eq:fbal} satisfies all the conditions required to efficiently minimize it using the method we described in Section~\ref{sec:mainloop}.
\begin{lemma}\label{lem:fballem}
The function $f$ is convex, second-order robust with respect to $\ell_\infty$,
and its Hessian is \SDDM.
\end{lemma}

The method we described in Section~\ref{sec:mainloop} depends on a promise concerning the point we initialize it with. Recall that in order to apply Theorem~\ref{thm:sor_sdd} we require an upper bound on the size of the $\ell_\infty$-ball containing the level set of the initial point. In order to provide good bounds, we regularize $f$. The description and effect of this regularization in captured in the following lemma.

\begin{lemma}\label{lem:balreg}
Suppose that there exists a point $x$ such that $f(x) \leq f_* + \eps^2\lA/24$, and $\| x
\|_{\infty} \leq B$. Then, the regularization of $f$ is defined as
\begin{equation}
\widetilde f(x) = f(x) + \frac{\eps^2 \lA}{48 n e^B} \sum_{i=1}^n  (e^{x_i} + e^{-x_i})
\end{equation} and satisfies the following properties:
\begin{enumerate}
\item $\widetilde{f}$ is second-order robust with respect to $\ell_\infty$ and has a {\SDDM} Hessian,
\item $f(x)\leq \widetilde f(x)$, and if $\widetilde{x}^*$ is the minimizer of $\widetilde{f}$, then
$\widetilde{f}(\widetilde{x}^*) \leq f(x^*) + \epsilon^2 \lA /24$,
\item for all $y$ such that $\widetilde{f}(y) \leq \widetilde{f}(0)$, $\|y - x^*\|_\infty = O(B \log (n \wA/\epsilon))$.
\end{enumerate}
\end{lemma}
The details of the lemma are identical to Lemma~\ref{lem:scalreg}, and we therefore omit the
proof.
In particular, this lemma implies that approximately optimizing the regularized function will still
produce an approximately balanced matrix.

Theorem~\ref{thm:balancing} then follows by applying Theorem~\ref{thm:sor_sdd} to the regularized function defined in Lemma~\ref{lem:balreg}, and combining it with the error guarantee of Lemma~\ref{lem:balancing_requirement}.
The details are presented in Section~\ref{sec:mainbal} of the Appendix. Similarly to the case of the scaling problem, we don't need to know any a priori bound on $B$. Just trying increasingly larger
value of $B$ (i.e., doubling our guess at each iteration) is sufficient.

\subsubsection{Bounding the Condition Number of the Optimal Balancing}
As we saw above, the running time given by Theorem~\ref{thm:balancing} depends
logarithmically on $\kappa(\DD^*)$, where $\DD^*$ is the matrix that achieves the optimal
balancing. While, in general $\kappa(\DD^*)$ can be exponentially large (and therefore we
might be better off running the interior point method described in Section~\ref{sec:ipm}), tighter connectivity of the graph implies better bounds:

\begin{lemma}\label{lem:baldiam}
Let $\AA \in \R^{n\times n}$ be a nonnegative matrix. Suppose that the graph with adjacency
matrix $\AA$ is strongly connected, and every vertex can reach every other vertex within at
most $k$ hops.  Then the matrix $\DD^*$ that perfectly balances $\AA$ has $\log \kappa(\DD^*)
= O(k \log \wA )$.
\end{lemma}
The proof of the lemma is in Section~\ref{sec:pfofbaldiam} of the Appendix and it yields the following upper bound on the value of $\kappa(\DD^*)$.
\begin{corollary}\label{cor:cond}
If  $\AA$ is a balanceable matrix, and $\DD^*$ perfectly balances it, then $\log \kappa(\DD^*) = O(n \ln \wA)$. If $\AA$ is strictly positive, then $\log \kappa(\DD^*) = O(\ln \wA)$.
\end{corollary}

\subsection{Discussion of Numerical Precision Aspects}
The exposition of the analysis so far is under the assumption of exact arithmetic. However,
our algorithms do in fact tolerate finite fixed-point precision on the scale of the
natural parameters of the problem ($n$, $\epsilon$, $\uA$ and $\wA$). It is therefore
sufficient to use a number of bits that is logarithmic in the input parameters of the problem.

Between iterations, we store a fixed-point representation of the variables $x_i$.  These are,
by construction, bounded by the parameter $R_\infty$.  It is important (at least if using
fixed-point rather than floating-point) that we are storing the $x_i$ rather than the actual
scalings $e^{x_i}$.

When iterating, we first determine the post-scaling elements of the matrix.  These can also
simply be stored in fixed-point--i.e. up to additive error. Note that this rounding could
completely eliminate very small entries of the matrix.  This representation then gives us the
gradient and Hessian of the problem up to additive error. To make the additive error
polynomially small only a logarithmic number of bits are needed, because the entries
of the scaled matrix can never be more than polynomial (in the natural parameters mentioned).
This follows from the fact that the objective function, which includes the sum of all the entries,
cannot increase.

Finally, there is an polynomially small absolute lower bound on the eigenvalues of the Hessian,
simply from the regularizer itself.  This ensures that additive error to the gradient and
Hessian can only affect the function value improvement by a polynomially larger amount,
and ensures the stability of the $k$-oracle algorithm.
Thus polynomially small error is sufficient, requiring only logarithmically many bits.

\subsection{Matrix Scaling and Balancing as Nonlinear Flow Problems}
An intriguing property of the matrix scaling and matrix balancing problems is that they both can be phrased as an instance of a more general problem. This problem can be seen as generalization of the electrical flow problem. That is, the problem of finding a potential-induced flow that routes a fixed demand in the case when Ohm's Law, i.e., the relationship between the potential difference on a given edge and the flow flowing through it is exponential instead of being linear.  (See ~\cite{duffin1947nonlinear} for a comprehensive treatment of such nonlinear networks.) To see this, given a weighted directed graph $G = (V, E, w)$ let us  define the \textit{edge-vertex} incidence matrix $\BB$ being an $n \times m$ matrix with rows indexed by edges and columns indexed by vertices such that
$$
\BB_{v,e} = \begin{cases}
1 & \textnormal{if } e = (v, u) \in E\,{,} \\
-1 & \textnormal{if } e = (u, v) \in E\,{,} \\
0 & \textnormal{otherwise.}
\end{cases}
$$

Using this matrix we define the nonlinear operator $\mathcal{L}$ as follows.
\begin{defn}\label{def:nonlin}
Let $G = (V,E,w)$ be a directed graph with vertex-edge incidence matrix $\BB$, and let $\mvar{W} = \mdiag(w)$. We define the operator $\mathcal{L}$ associated with $G$ as
\begin{equation}\label{eq:nonlin}
\mathcal{L}(x) = \BB^\top \mvar{W} \exp(\BB x ) \,{.}
\end{equation}
\end{defn}
This can be seen as a nonlinear generalization of the Laplacian operator, which is a linear operator defined as $\LL = \BB^\top \mvar{W} \BB$. There is extensive literature on solving Laplacian linear systems~\cite{SpielmanT04,KoutisMP10,KoutisMP11,KelnerOSZ13,CohenKMPPRX14,KyngLPSS16,KyngS16}. We argue that our framework can be used to solve systems of the form
\begin{equation}
\mathcal{L}(x) = d\,{.}
\end{equation}
This can be seen as finding vertex potentials $x$ which induce a flow vector $f$:
$$
f_{uv} = w_{uv} \cdot e^{x_u - x_v}\,{,} \quad \textnormal{for all $(u,v) \in E$},
$$
such that $f$ routes a given demand $d$. This should be contrasted with the case of electrical flows where the flow is induced as $f_{uv}=w_{uv}(x_u-x_v)$. As it turns out, the solution to the system $\mathcal{L}(x)=d$ is the minimizer of a function similar to those defined in Equations~\ref{eq:scaling_fn} and~\ref{eq:fbal}. More precisely:
\begin{lemma}
Let $G = (V,E,w)$ be a directed graph with nonnegative weights,  let  $\AA$ be its adjacency matrix, and let $\mathcal{L}$ be the operator associated with  $G$, as defined as in Equation~\ref{eq:nonlin}.
Consider the function $f$ defined as
\begin{equation}
f(x) = \sum_{1\leq i, j \leq n} \AA_{ij} e^{x_i - x_j} - \sum_{1\leq i \leq n} d_i x_i\,{.}
\end{equation}
Then $f$ has a minimizer  $x$ if and only if it is the solution to the system $\mathcal{L}(x) = d$.
\end{lemma}
The proof follows directly from writing optimality conditions for $f$, noting that the condition that $\nabla f(x) = 0$ is equivalent to $\mathcal{L}(x) = d$. Similarly to Theorem~\ref{thm:escaling} and Theorem~\ref{thm:ebalancing}, we can provide conditions on function value error to bound the error $\|d - \mathcal{L}(x)\|_2$. Also, in order to obtain a good running time, we require regularizing $f$ in a manner similar to the regularization applied in Lemmas~\ref{lem:scalreg} and~\ref{lem:balreg}. Note that in the case of the scaling and balancing problems, since we require problem specific error guarantees, the regularization needs to be customized accordingly. 

Finally, we state without proof that balancing and scaling are instances of solving $\mathcal{L}(x) = d$.
\begin{obs}
Let $\AA \in \R^{n \times n}$ be a balanceable nonnegative matrix. Let $\mathcal{L}$ be the nonlinear operator associated with the graph with adjacency matrix $\AA$.
Then the solution $x$ to $\mathcal{L}(x) = 0$ yields a balancing $\mdiag(\exp(x))$.
\end{obs}

\begin{obs}
Let $\AA \in \R^{n \times n}$ be a $(r,c)$-scalable nonnegative matrix. Let $\mathcal{L}$ be the nonlinear operator associated with the graph with adjacency matrix
$\begin{bmatrix} 
\mzero & \AA  \\
\mzero & \mzero
\end{bmatrix}$. Then the solution $z = (x,y)^\top$, with $x,y \in \R^n$ to $\mathcal{L}(z) = (r, -c)^\top$ yields a $(r,c)$-scaling $(\mdiag(\exp(x)), \mdiag(\exp(y))$.
\end{obs}

%% file: schur.tex
\section{Implementing an $O(\log n)$-Oracle in Nearly Linear Time}\label{sec:nearlylineartime}

In Section~\ref{sec:scb} we reduced the balancing and scaling problems to the approximate
minimization of second-order robust functions with respect to the $\ell_\infty$ norm. All
that is left to have a complete algorithm, we need a fast procedure to implement a $k$-oracle
as in Definition~\ref{def:oracle}.
Namely, show how to construct an $O(\log n)$-oracle for the problem,
\begin{equation}
\label{eq:problem}
\min_{\|x\|_\infty \leq 1} x^\top\MM x + \langle b, x\rangle\,{,}
\end{equation}
where $\MM$ is an {\SDDM} matrix. For this section, whenever we say that a matrix is SDD we
will also imply that the off-diagonal entries are nonpositive.

One possible approach, is to use standard convex optimization reductions to turn this problem
into the minimization of the maximum of an $\ell_\infty$ norm and an $\ell_2$ norm subject to
linear constraints. This problem can be solved in time $\Otil(mn^{1/3})$ using the
multiplicative weights framework as applied in \cite{CKMST,CMMP13}. The resulting algorithm
for implementing the $k$-oracle would take time $\Otil(m+n^{4/3})$, by taking advantage of
spectral sparsification algorithms \cite{SS11,LeeS15,LeeS17}.
Instead, we will come up with a faster algorithm.

Our approach, based on the Lee-Peng-Spielman solver \cite{LeePS15}, is to identify large sets of
vertices where the problem is ``easy'' to solve and then deal with the rest of the
graph (reduced in size) recursively. The particular notion of ``easy'' we are
going to use, is that of strong diagonal dominance.
\begin{definition}\label{def:aSDD}
A matrix $\MM$ is \emph{$\alpha$-strongly diagonally dominant ($\alpha$-SDD)},
if for all $i$
$$\MM_{ii} \geq (1+\alpha) \sum_{j\neq i} |\MM_{ij}|.$$
\end{definition}
The reason that such matrices enable us to solve the corresponding problems
fast is that they can be well-approximated by a diagonal matrix.
\begin{lemma}\label{lem:diagapprox}
Every $\alpha$-SDD matrix $\MM$, with diagonal $\mdiag(\MM)$, satisfies
$$\left (1 - \frac{1}{1+\alpha} \right )\mdiag(\MM) \pleq \MM \pleq \left(1+\frac{1}{1+\alpha}\right) \mdiag(\MM).$$
\end{lemma}
\begin{proof}
This follows from the fact that
\begin{equation*}
\begin{bmatrix} -1 & 0 \\ 0 & -1 \end{bmatrix} \preceq \begin{bmatrix} 0 & 1 \\ 1 & 0 \end{bmatrix} \preceq \begin{bmatrix} 1 & 0 \\ 0 & 1 \end{bmatrix}
\end{equation*}
Applying this to each off-diagonal entry, the off-diagonal part of the matrix will be
bounded between diagonal matrices; by the $\alpha$-SDD property these can be bounded by
$\pm \frac{1}{1+\alpha} \mdiag(\MM)$.
\end{proof}
In our context, problems in the form of Equation~\ref{eq:problem}, where $\MM$ is
an $\alpha$-SDD matrix for some $\alpha \geq \Omega(1)$, can be turned into well
conditioned quadratic minimization problems for which we can apply standard linearly
convergent algorithms. 
For a more detailed description and analysis of such algorithms can be found in \cite{nesterov1998introductory}.

\begin{lemma}
There is an algorithm $\textsc{FastSolve}$, that given 
an $\Omega(1)$-SDD matrix $\MM$, and $\epsilon > 0$, returns a point
$\widetilde x$, such that $\|\widetilde x\|_\infty\leq 2$, and
$$\widetilde x^\top \MM \widetilde x + \langle \widetilde x,b\rangle
\leq (1-\epsilon) \min_{\|x\|_\infty\leq 2} x^\top \MM x + \langle x,b\rangle$$
in time $O(m\log(1/\epsilon))$, where $m$ is the number of nonzero entries of $\MM$.
\end{lemma}
\begin{proof}
By Lemma~\ref{lem:diagapprox}, there is some diagonal matrix $\DD$ such that 
$$\left (1 - \frac{1}{1+\alpha} \right )\DD \pleq \MM \pleq \left(1+\frac{1}{1+\alpha}\right) \DD.$$
By applying the transformation $x = \DD^{-1/2}z$, the problem
becomes
$$\min_{\|\DD^{-1/2}z\|_\infty\leq 2} h(z) = z^\top \DD^{-1/2}\MM\DD^{-1/2}z + \langle
\DD^{-1/2}b, x\rangle.$$
We will apply proximal gradient descent, defined by the sequence $z^{(0)}=0$ and
$$z^{(t+1)} = \arg\min_{\|\DD^{-1/2}z\|_\infty\leq 2} \left\{\langle \nabla
h(z^{(t)}),z\rangle + \left(1+\frac{1}{1+\alpha}\right)
\frac{\|z-z^{(t)}\|_2^2}{2}\right\}.$$
Computing $z^{(t+1)}$ from $z^{(t)}$ corresponds to computing
$$z' = z^{(t)} - \frac{1+\alpha}{2+\alpha}\nabla h(z^{(t)}) =
    z^{(t)} - \frac{1+\alpha}{2+\alpha} \DD^{-1/2}\MM\DD^{-1/2}z^{(t)},$$
and projecting it to the space $\|\DD^{-1/2}z\|_\infty\leq 2$, by simply trunctating any coordinates
exceeding the bounds.  We can clearly implement each iteration in linear time in the number
of nonzeros of $\MM$. Since the condition number of the function is at most
$(1+2/\alpha)=O(1)$, such a step will imply that
$$h(z^{(t)}) - h(z^{(t+1)}) \geq \frac{1}{O(1)}(h(z^{(t)}) - h(z^*),$$
and thus inductively,
$$h(z^{(t)}) - h(z^*) \leq \left(1-\frac{1}{O(1)}\right)^t (h(z^{(0)}) - h(z^*)).$$
Therefore, after $O(\log(1/\epsilon))$ steps we will have a point with
$h(z^{(t)}) - h(z^*) \leq \epsilon (h(z^{(0)}) - h(z^*)).$
The fact that $h(0) = 0$ concludes the proof.
\end{proof}

An even simpler case is when the matrix is of size 1, in which case the problem can be \emph{exactly}
solved in constant time:
\begin{lemma}
There is an algorithm $\textsc{TrivialSolve}$, that given a 1 by 1 matrix $\MM$
returns an $x$ optimizing $x^\top\MM x + \langle b, x\rangle$ over the interval $[-1, 1]$.
\end{lemma}
\begin{proof}
By convexity, there must exist an optimal $x$ that is either one of the two endpoints
of the interval, or the unique global optimum of the function over the line.  One may
simply check all candidates and return the best value.
\end{proof}

A key insight of \cite{LeePS15} is that one can find $\Omega(1)$-SDD submatrices of $\MM$ of
size $\Omega(n)$. We denote such a subset by $F$ and $V\setminus F$ by $C$. To ensure that solving
the problem for $x_F$ will not interfere with our solution $x_C$ we map a solution
$\hat x_C$ supported only on coordinates of $C$ to a solution $x_C$ through a
linear mapping $\PP$. If $\PP$ were the energy minimizing extension of
voltages on $C$ to voltages on $V$,
$$(\PP\hat x_C)_{F}=\MM_{[F,F]}^{-1}\MM_{[F,C]}\hat x_C,$$
we would have that $x_F$ and $x_C$ are $\MM$-orthogonal,
since $ x_F^\top\MM \PP \hat x_C = 0$. Then, optimizing over $\hat x_C$
involves the quadratic $\PP^\top \MM \PP$ which is exactly equal to
$\MM_{[C,C]}-\MM_{[C,F]} \MM_{[F,F]}^{-1} \MM_{[F,C]} = \Sc(\MM, F)$.
Applying this proccess recursively leads to the notion of vertex sparsifier
chains that we will heavily rely on.

\begin{definition}[Definition 5.7 of \cite{LeePS15}]\label{def:sparschain}
For any $SDD$ matrix $\MM^{(0)}$, a vertex sparsifier chain of $\MM^{(0)}$ with
parameters $\alpha_i\geq 4$ and $1/2\geq \eps_i > 0$, is a sequence of matrices
and subsets $(\MM^{(1)},\ldots,\MM^{(d)};F_1,\ldots, F_{d-1})$ such that
\begin{enumerate}
\item $\MM^{(1)}\approx_{\eps_0} \MM^{(0)}$,
\item $\MM^{(i+1)}\approx_{\eps_i} \Sc(\MM^{(i)}, F_i)$,
\item $\MM_{[F_i,F_i]}^{(i)}$ is $\alpha_i$-strongly diagonally dominant, and
\item $\MM^{(d)}$ has size 1.
\end{enumerate}
\end{definition}
Note that this last requirement is slightly different from \cite{LeePS15}:
we require the chain to end with size 1, rather than just being constant.
However, the chain construction from \cite{LeePS15} immediately extends to
this requirement (they presumably proposed stopping early because it is
a simple optimization that would likely be valuable in any implementation).

To be able to reason about the approximation guarantees of the chain as a whole
we will use an error-quantifying definition.
\begin{definition}[Definition 5.9 of \cite{LeePS15}]\label{def:chainwork}
An $\epsilon$-vertex sparsifier chain of an SDD matrix $\MM^{(0)}$ of work $W$, is a
vertex sparsifier chain of $\MM^{(0)}$ with
parameters $\alpha_i\geq 4$ and $1/2\geq \eps_i > 0$ that satisfies
\begin{enumerate}
\item $2\sum_{i=0}^{d-1}\eps_i \leq \epsilon$,
\item $\sum_{i=0}^{d-1} m_i \log_{\alpha_i}\eps_i^{-1}\leq W$, where $m_i$ is the
number of nonzeros in $L^{(i)}$.
\end{enumerate}
\end{definition}

Finally, the construction of such chains, as well as their error guarantees have
been already analyzed in $\cite{LeePS15}$ and can be used in a black-box
manner.\\

\begin{theorem}[Theorem 5.10 of \cite{LeePS15}]\label{thm:sparschain}
Every SDD matrix $\MM$ of dimension $n$ has a $\delta$-vertex sparsifier chain of work
$O(n)$ and $d\leq O(\log n)$, for any constant $0 < \delta \leq 1$. Such a chain can be
constructed in time, $\Otil(m)$.
\end{theorem}
We note that Theorem~\ref{thm:sparschain} was stated for $\delta = 1$, but it is
straightforward to modify to proof without changing the work or the length of
the chain by more than a constant factor.

Since we cannot exactly compute the energy minizing mapping $\PP$, we will define
an approximate mapping that suffices for our purposes.

\begin{definition}\label{def:mapping}
A linear mapping $\PPtil$ is an \emph{$\epsilon$-approximate voltage extension from $C$
to $V$ according to $L$} if for any $\hat x_C\in\R^{|C|}$,
\begin{enumerate}
\item $\| (\PPtil- \PP) \hat x_C \|_\MM \leq \epsilon \| \PP \hat x_C \|_\MM$,
\item $\PPtil$ is the identity on coordinates in $C$
\item the coordinates of $\PPtil \hat x_C$ are convex combination of the coordinates
of $\hat x_C$ and 0.
\end{enumerate}
where $\PP$ is the energy minimizing extension.
\end{definition}

We will construct such a mapping through a simple averaging scheme. First we set
the voltage of every vertex in $F$ to be the weighted average of its neighbors
in $C$. Then at every step we replace its voltage by the weighted average of
\emph{all} its neighbors. (Here, excess diagonal is treated as an edge to
a vertex with voltage 0.) We do so for $O(\log(1/\epsilon))$ iterations. We
formally describe the procedure in Figure~\ref{fig:mapping}.

It is easy to see that all steps of the algorithm are linear maps, and we can therefore also implement its transpose.
\begin{figure}
\frame{
\begin{minipage}{\textwidth}
\vspace{10pt}
\hspace{5pt}
$\textsc{ApproxMapping}(\MM,C,\epsilon) \hat x_C$

\begin{enumerate}
\item $T \gets \frac{\log \left ( \sqrt{1+\frac{1}{\alpha}} \right ) + \log \left ( \frac{1}{\epsilon} \right )}{\log(1+\alpha)}$
\item
$x^{(0)}_C \gets \hat x_C$

$x^{(0)}_F \gets -\EE^{-1} \MM_{[F,C]} \hat x_C$, $\EE$ is external degree matrix $\EE_{ii} = \MM_{ii}-\sum_{j \in F,j \neq i} |\MM_{ij}|$
\item For $t\gets 1,\ldots, T$: \\
$x^{(t)}_C \gets x^{(t-1)}_C$

$x^{(t)}_F \gets -\mdiag(\MM_{[F,F]})^{-1} \MM_{[F,F \cup C]} x^{(t-1)}$

\item Return $x^{(T)}$
\end{enumerate}
\vspace{5pt}
\end{minipage}
}
\caption{Implementation of an $\eps$-approximate mapping}
\label{fig:mapping}
\end{figure}

\begin{lemma}
For any SDD matrix $\MM$, given an $\Omega(1)$-SDD subset $F$ and some $\epsilon
>0$ one can apply an $\epsilon$-approximate voltage extension mapping in time
$O(m\log(1/\epsilon))$.
\end{lemma}
\begin{proof}
The linearity of the mapping, properties 2 and 3, as well as the runtime claimed
hold inductively by the construction of the mapping.
We will now prove property 1, namely, if $P$ is the true energy minimizing
mapping and $\widetilde P$ is our approximate mapping, for any $\hat
x\in\R^{|C|}$, $\| (\widetilde P - P) \hat x_C \|_{\MM} \leq \epsilon \| P \hat x_C \|_{\MM}$.

First, we will bound the error of $x^{(0)}$.  We define $v^{(t)} = x^{(t)} - P \hat x_C$,
which is 0 outside $F$ by construction.  We will use the notation $w_{ij} = -{\MM}_{ij}$--i.e.
the edge weight between $i$ and $j$, and $w_{i\varnothing} = \MM_{ii}-\sum_{j \neq i} |\MM_{ij}|$.
Here, $w_{i\varnothing}$ accounts for ``excess diagonal'' of $\MM$, which we treat as an edge
to a ``virtual vertex'' $\varnothing$.  We will also define $(\hat x_C)_\varnothing = 0$.
Now, we have
\begin{align*}
\| P \hat x_C \|_{\MM}^2 &\geq \sum_{i\in F, j\in C\cup\{\varnothing\}} w_{ij} ((P \hat x_C)_i - (\hat x_C)_j)^2 \\
&\geq \sum_{i\in F} \left ( \sum_{j \in C\cup\{\varnothing\}} w_{ij} \right ) (v^{(0)}_i)^2 \\
&\geq \frac{\alpha}{1+\alpha} \| v^{(0)} \|_{\mdiag(\MM)}^2.
\end{align*}
Rearranging gives
$\| v^{(0)} \|_{\mdiag(\MM)} \leq \sqrt{1+\frac{1}{\alpha}} \| P \hat x_C \|_{\MM}$.

Next, we note that $v^{(t)} = \mdiag(\MM)^{-1} (\II-\MM)_{[F,F]} v^{(t-1)}$.
This follows from the fact that
\begin{equation*}
-\mdiag(\MM_{[F,F]})^{-1} \MM_{[F,F \cup C]} (P \hat x_C) = 0
\end{equation*}
since
$\MM (P \hat x_C)$ is 0 on $F$.  Applying
Lemma~\ref{lem:diagapprox}, we get
\begin{equation*}
\| \mdiag(\MM)^{-1/2} (\II-\MM)_{[F,F]} \mdiag(\MM)^{-1/2} \| \leq \frac{1}{1+\alpha}.
\end{equation*}
This implies that
\begin{equation*}
\| v^{(t)} \|_{\mdiag(\MM)} \leq \frac{1}{1+\alpha} \| v^{(t-1)} \|_{\mdiag(\MM)}.
\end{equation*}
By induction, we have
\begin{equation*}
\| v^{(t)} \|_{\mdiag(\MM)} \leq \frac{\sqrt{1+\frac{1}{\alpha}}}{(1+\alpha)^t} \| P \hat x_C \|_{\MM}.
\end{equation*}
Finally, using Lemma~\ref{lem:diagapprox} again, we get
\begin{equation*}
\| v^{(t)} \|_{\mdiag(\MM)} \leq \frac{\sqrt{1+\frac{2}{\alpha}}}{(1+\alpha)^t} \| P \hat x_C \|_{\MM}.
\end{equation*}
\end{proof}

Given such a mapping, we can uniquely express any voltage vector as $x = x_C +
x_F$, where $x_C = \PPtil \hat x_C$ (i.e. it is in the image of $\widetilde P$)
and $x_F$ is supported on $F$.  By the convex combination property of $\PPtil$, we have $\|
x_C \|_\infty \leq \| \hat x_C \|_\infty \leq \| x \|_\infty$; since $x_F = x -
x_C$, by the triangle inequality we have $\| x_F \|_\infty \leq 2 \| x
\|_\infty$. The domain $\| x \|_\infty \leq 1$ is therefore contained in $\|
\hat x_C \|_\infty \leq 1, \| x_F \|_\infty \leq 2$. Moreover, any point for
which $\|\hat x_C\|_\infty \leq k$, $\|x_F\|_\infty \leq 2$, corresponds to
a point $x = \PPtil \hat x_C + x_F$ with $\|x\|_\infty \leq k+2$.

Having expressed all of the components of our approach, stating the algorithm is
simple. Given the decomposition of the problem the vertex sparsifier chain
provides, we will solve the smallest problem and then iteratively combine it
with the solution of the submatrices along the chain. The algorithm is formally
described in Figure~\ref{fig:chain}, and the main claim in
Theorem~\ref{thm:subproblem}.

\begin{figure}
\frame{
\begin{minipage}{\textwidth}
\vspace{10pt}
\hspace{5pt}
$\textsc{OptimizeChain}((\MM^{(1)},\ldots,\MM^{(d)};F_1,\ldots,
F_{d-1};\eps_0,\ldots \eps_{d-1}), b)$

\begin{enumerate}
\item
$b^{(1)} \gets b/e^{\eps_0}$
\item For $i\gets 1,\ldots, d-1$
\begin{enumerate}
\item $\PPtil^{(i)}\gets \textsc{ApproxMapping}(\MM^{(i)},F_i,\eps_i)$
\item $b^{(i+1)} \gets (\PPtil^{(i)})^\top b^{(i)}/(e^{\eps_i}(1 + \eps_i +
\eps_i^2))$
\end{enumerate}
\item $x^{(d)}\gets\textsc{TrivialSolve}(\MM^{(d)},b^{(d)})$
\item For $i\gets d-1,\ldots, 1$
\begin{enumerate}
\item $x_C^{(i)} \gets \PPtil^{(i)} x^{(i+1)}$
\item $x_F^{(i)} \gets \textsc{FastSolve}(\MM_{[F_i,F_i]}^{(i)}, b_{F_i}^{(i)},
\eps_i)$
\item $x^{(i)} \gets x_C^{(i)} + x_F^{(i)}$
\end{enumerate}
\item Return $x^{(1)}$
\end{enumerate}
\vspace{5pt}
\end{minipage}
}
\caption{Optimizing a vertex sparsifier chain}
\label{fig:chain}
\end{figure}

To facilitate the analysis we first state the following decoupling lemma.
\begin{lemma}\label{lem:decouple}
Consider an SDD matrix $\MM$, a partition of its columns $(F,C)$ and some $0 <
\epsilon \leq 1/2$. Let
$\PPtil$ be an $\epsilon$-approximate voltage extension from $C$ to $(F,C)$ as
define in Definition~\ref{def:mapping}. Then
\begin{align*}
(\PPtil \hat x_C + x_F)^\top \MM (\PPtil \hat x_C + x_F) &\leq
(1+\epsilon+\epsilon^2) (\PP \hat x_C)^\top \MM (\PP \hat x_C) + (1+\epsilon)
x_F^\top \MM x_F\\
 (\PPtil \hat x_C + x_F)^\top \MM (\PPtil \hat x_C + x_F) &\geq (1-\epsilon+\epsilon^2) (\PP \hat x_C)^\top \MM (\PP \hat x_C) + (1-\epsilon)
x_F^\top \MM x_F .
\end{align*}
\end{lemma}
\begin{proof}
Since $\PP$ is the true energy minimizing extension, we know that $\MM (\PP \hat
x_C)$, will be zero on the coordinates of $F$. Since $(\PPtil - \PP)\hat x_C +
x_F$ is supported on $F$ (property 2 of Definition~\ref{def:mapping}), we can expand
\begin{align*}
(\PPtil \hat x_C + &x_F)^\top \MM (\PPtil \hat x_C + x_F) =\\
&= (\PP \hat x_C +
(\PPtil-\PP) \hat x_C + x_F)^\top \MM (\PP \hat x_C + (\PPtil-\PP) \hat x_C + x_F) \\
&= (\PP \hat x_C)^\top \MM (\PP \hat x_C) + ((\PPtil-\PP) \hat x_C + x_F)^\top \MM
((\PPtil-\PP) \hat x_C + x_F) \\
&= ((\PP \hat x_C)^\top \MM (\PP \hat x_C) + x_F^\top \MM x_F + 2 ((\PPtil-\PP) \hat
x_C)^\top \MM
x_F + ((\PPtil-\PP) \hat x_C)^\top \MM ((\PPtil-\PP) \hat x_C).
\end{align*}
The first property of $\PPtil$ in Definition~\ref{def:mapping}, implies that
$$((\PPtil-\PP) \hat x_C)^\top \MM ((\PPtil-\PP) \hat x_C) \leq \epsilon^2 (\PP\hat
x_C)^\top\MM(\PP\hat x_C).$$
We can upper bound the contribution of the cross-term as
\begin{align*}
2 ((\PPtil-\PP) \hat x_C)^\top \MM x_F 
    &\leq \frac{1}{\epsilon} ((\PPtil-\PP) \hat x_C)^\top \MM ((\PPtil-\PP) \hat x_C) +
        \epsilon x_F^\top \MM x_F \\
    &\leq \epsilon (\PP \hat x_C)^\top \MM (\PP \hat x_C) +
        \epsilon x_F^\top \MM x_F. 
\end{align*}
Similarly, we can lower bound the contribution of the cross-term as
\begin{align*}
2 ((\PPtil-\PP) \hat x_C)^\top \MM x_F 
    &\geq -\frac{1}{\epsilon} ((\PPtil-\PP) \hat x_C)^\top \MM ((\PPtil-\PP) \hat x_C) 
        -\epsilon x_F^\top \MM x_F \\
    &\geq -\epsilon (\PP \hat x_C)^\top \MM (\PP \hat x_C) 
        -\epsilon x_F^\top \MM x_F. 
\end{align*}
Combining these inequalities and rearranging terms concludes the proof.
\end{proof}
We can now use this lemma to prove that by decoupling the problem and
using approximate Schur complements does not reduce the quality of the solution
by more than a constant.
\begin{theorem}\label{thm:subproblem}
Algorithm $\textsc{OptimizeChain}$ implements a $O(\log n)$-oracle,
and runs in time $\Otil(m)$.
\end{theorem}
\begin{proof}
By construction and the triangle inequality we have that
$$\|x^{(0)}\|_\infty \leq \|x^{(d)}\|_\infty + \sum_{i=1}^d
\|x_F^{(i)}\|_\infty \leq 1 + 2d \leq O(\log n).$$

We will define the following functions to reason about our approximation guarantees:
\begin{align*}
H_i(x)&= x^\top \MM^{(i)} x + \langle b^{(i)}, x \rangle \\
H'_i(x) &=(1+\epsilon_i+\epsilon_i^2) x^\top \Sc(\MM^{(i)},F_i)
    x + \langle  b^{(i)}, \PPtil^{(i)}x \rangle\\
H''_i(x) &=e^{\eps_i}(1+\epsilon_i+\epsilon_i^2) x^\top \MM^{(i+1)}
    x + \langle  b^{(i)}, \PPtil^{(i)}x \rangle\\
G_i(x) &= (1+\eps_i) x^\top\MM_{[F_i,F_i]}^{(i)}x + \langle b_F^{(i)},x\rangle
\end{align*}
We can now derive the following facts about these functions:
\begin{enumerate}
\item Directly applying Lemma~\ref{lem:decouple}, for any  $\hat x_C$, $x_F$,
$$H_i(\PPtil\hat x_C + x_F) \leq H'_i(\hat x_C) + G_i(x_F)\,{.}$$
\item  Using the fact that any $x$ with $\|x\|_\infty\leq 1$ can be written as $\PPtil
\hat x_C + x_F$ for $\|\hat x_C\|_\infty\leq 1$ and $\|x_F\|_\infty\leq 2$, and
Lemma~\ref{lem:decouple},
$$\min_{\|\hat x_C\|_\infty\leq 1} H'_i(\hat x_C) + \min_{\|x_F\|_\infty\leq 2}
G_i(x_F)\leq \left(\frac{1-\eps_i}{1+\eps_i}\right)\min_{\|x\|_\infty\leq 1} H_i(x)\,{.}$$
\item By the definition of the vertex sparsifier chain, $\MM^{(i+1)}
\approx_{\eps_i} \Sc(\MM^{(i)}, F_i)$, implying for any $x$
$$H_i'(x) \leq H''_i(x)\,{.}$$
\item Again, from the fact that $\MM^{(i+1)} \approx_{\eps_i} \Sc(\MM^{(i)}, F_i)$,
$$\min_{\|x\|_\infty\leq 1} H_i''(x) \leq e^{-2\eps_i} \min_{\|x\|_\infty\leq 1}
H'_i(x)\,{.}$$
\item  By the definition of $b^{(i)}$,
$$H''_i(x) = (1+\eps_i) (1+\eps_i+\eps_i^2) H_{i+1}(x)\,{.}$$
\end{enumerate}

We are going to combine these facts to show that for any $i\in[1,d]$,
$$H_i(x^{(i)})\leq e^{\sum_{j=i}^{d-1} -5\eps_j}\min_{\|x\|_\infty \leq 1} H_i(x).$$
We procceed by induction from $d$ to $1$. Recall that $\eps_i\leq 1/2$.\\
For the case of $i=d$ it trivially holds by the guarantee of
$\textsc{TrivialSolve}$:
$$H_d(x^{(d)}) \leq \min_{\| x \|_\infty \leq 1} H_d(x)$$
Assuming that it holds for any $j > i$, by the guarantees of
$\textsc{FastSolve}$:
\begin{align*}
H_i(x^{(i)}) &\leq H_i'(x^{(i+1)}) + G_i(x_F^{(i)})\\
    &\leq H_i''(x^{(i+1)}) + G_i(x_F^{(i)})\\
    &= \frac{1}{e^{\eps_i}(1+\eps_i+\eps_i^2)}H_{i+1}(x^{(i+1)}) + G_i(x_F^{(i)})\\
    &\leq \frac{e^{\sum_{j=i+1}^{d-1}-5\eps_j}}{e^{\eps_i}(1+\eps_i+\eps_i^2)}\min_{\|x\|_\infty \leq 1} H_{i+1}(x) +
    (1-\eps_i)\min_{\|x\|_\infty \leq 2} G_i(x)\\
    &\leq e^{\sum_{j=i+1}^{d-1}-5\eps_j}\min_{\|x\|_\infty \leq 1}H''_i(x) + e^{-2\eps_i}\min_{\|x\|_\infty \leq 2}G_i(x)\\
    &\leq e^{\sum_{j=i+1}^{d-1}-5\eps_j}e^{-2\eps_i}\min_{\|x\|_\infty \leq 1} H'_i(x) + e^{-2\eps_i}\min_{\|x\|_\infty \leq 2} G_i(x)\\
    &\leq e^{\sum_{j=i+1}^{d-1}-5\eps_j}e^{-2\eps_i}\epsilon^{-3\eps_i}\min_{\|x\|_\infty \leq 1}H_i(x)\\
    &= e^{\sum_{j=i}^{d-1}-5\eps_j}\min_{\|x\|_\infty \leq 1} H_i(x).
\end{align*}

Finally, we can similarly argue about $H_0(x^{(1)})$: it is at most $e^{\eps_0} H_1(x^{(1)})$
while $\min_{\|x\|_\infty \leq 1} H_1(x) \leq e^{-3 \eps_0} \min_{\|x\|_\infty \leq 1} H_0(x)$.
By the guarantees of the vertex sparsifier chain, we know that $\sum_{i=0}^{d-1}
\eps_i \leq \delta$ for some constant $\delta$ of our choice. By choosing the
right constant we can ensure that our multiplicative error is less than $1/2$.

In order to bound the runtime we notice that for every $i$ we need to compute
$\textsc{ApproxMapping}$ and $\textsc{FastSolve}$ which both take
time $O(m_i\log(1/\eps_i))$. By the bound on the work $W$ of the chain we get
that applying the chain takes time $O(n)$, while constructing it takes time $\Otil(m)$.
\end{proof}

%% file: ipm.tex
\section{Matrix Scaling and Balancing with Exponential Cone Programming}
\label{sec:ipm}
The algorithm developed in the previous sections is essentially optimal in the
regime where the ratio between the scaling factors is relatively small (say
polynomial in $n$). Since there are matrices for which this ratio is
exponential, we develop a complementary algorithm with negligible runtime
dependence on this ratio, at the cost of a mild increase in the dependence on
$m$. The algorithm is based on interior point methods.

Although interior point methods would seem like a natural option for the
problems of matrix scaling and balancing, standard formulations require solving
linear systems involving various rescalings of the input matrix. A priori, it is
not clear whether these can be solved faster than matrix multiplication time.
However, it turns out that a somewhat nonstandard formulation requires solving
linear systems for more structured matrices. Particularly, we will see that
these matrices admit a decomposition involving only matrices that are easy to
invert (triangular matrices, solvable by back substitution, and SDD matrices
which can be tackled via a standard Laplacian solver). 
Notably, a similar observation was made by Daitch and Spielman~\cite{DS08}, in
the case of interior point methods applied to flow problems on graphs.
\cite{KK96}~also consider a formulation similar to ours for the matrix scaling
problem, however they don't
prove exact convergence bounds or state the algorithm rigorously. Moreover,
since nearly-linear SDD solvers where not known at the time, this algorithm
provides no benefit compared to other approaches.

The main result of this section is the following.
\begin{theorem} \label{thm:ipmsc}
Given a nonnegative matrix $\AA \in \R^{n\times n}$, one can:
\begin{enumerate}
\item compute an $\epsilon$-balancing in time $$\Otil(m^{3/2} \ln (\wA \epsilon^{-1}))\,{,}$$
\item if the matrix is almost $(r,c)$-scalable, compute a $\epsilon$-$(r,c)$-scaling in time $$\Otil(m^{3/2} \ln (\uA \epsilon^{-1})  )\,{.}$$ 
\end{enumerate}
\end{theorem}

This is as a matter of fact a consequence of the fact that  a specific class of functions, which capture both balancing and scaling, can be minimized efficiently. We capture this result in the following Theorem.

\begin{theorem}\label{thm:fopt}
Let $\AA \in \R^{n\times n}$ be a nonnegative matrix with $m$ nonzero entries,  let $f$ be the function
\begin{equation}\label{eq:exp0}\tag{F1}
f(x) = \sum_{(i,j) \in \supp(\AA)} \AA_{ij} e^{x_i - x_j} - \langle d, x \rangle  \,{,}
\end{equation}
 and let $B_x$ be a positive real number. There exists an algorithm which, for any $\epsilon > 0$, finds a vector $x$ such that $f(x) - f(x^{*(B_x)}) \leq \epsilon$ (where
 $x^{*(B_x)}$ is the optimum of $f$ over the region $\| x \|_\infty \leq B_x$) in time
$$\Otil \left( m^{3/2} \ln \left ( 2 + B_x + \frac{s_\AA}{\epsilon} + \frac{\Vert d \Vert_1}{\epsilon} \right ) \right)\,{.}$$
\end{theorem}

Using this result, one can then conclude the proof of Theorem~\ref{thm:ipmsc}.
\begin{proof}[Proof of Theorem~\ref{thm:ipmsc}]
For the balancing objective, we first decompose the nonzero entries of $\AA$
into strongly connected components. For each component, we will call
$\ref{thm:fopt}$ with $d = 0$. From Corollary~\ref{cor:cond} we have that $\Vert
x^* \Vert_\infty = O(n \ln \wA)$. Plugging this in, along with Lemma~\ref{lem:balancing_requirement}, we obtain a total running time of $\Otil(m^{3/2} \ln (\wA  \epsilon^{-1}))$.

For the $\epsilon$-$(r,c)$-scaling objective, we set $d = (r, -c)^\top$, and run
the interior point method on the matrix $\begin{bmatrix} \mzero & \AA \\
\mzero & \mzero \end{bmatrix}$. Lemma~\ref{lem:rcbounds} ensures that the
entries of the $(r,c)$ scalings exist within a polynomially bounded
$\ell_\infty$-ball.
Using Lemma~\ref{lem:scaling_accuracy}, this yields the conclusion.
\end{proof}

We prove Theorem~\ref{thm:fopt} by showing that an interior point method defined and analyzed by Nesterov~\cite{nesterov1998introductory} can be efficiently implemented. In order to do so, we require two components. The former involves providing a formulation for minimizing the function in~\ref{eq:exp0} for which the interior point method can produce an iterate that is close in value to optimum within a small number of iterations. The latter involves showing how to efficiently implement these iterations. Generally they involve solving a linear system; in our case, we show that such iterations can be executed by solving an SDD linear system to constant accuracy.

\subsection{Setting Up the Interior Point Method, and Bounding the Number of Steps}

In order to apply an interior point method, we first reformulate the problem in~\ref{eq:exp0} in an equivalent form. 

\begin{lemma}[Equivalent Formulation]
Let $B_x$ be a promise on the magnitude of the entries in the optimal solution of~\ref{eq:exp0}: $$\Vert x^* \Vert \leq B_x\,{.}$$ Also, let $$U = (s_{\AA} + \Vert d \Vert_1 B_x)\,{.}$$ Then the objective
\begin{equation}\label{eq:opt1}\tag{F2}
\begin{split}
\min_{(t,x) \in S} \langle \allones, t \rangle - \langle d, x \rangle  \quad \textnormal{where} \quad S =\big\{ (t,x) \in \R^m \times \R^n : \AA_{ij} e^{x_i - x_j} \leq t_{ij} \leq 3U \,\textnormal{, for all }(i,j) \in \supp(\AA), \\
-B_x \leq x_i \leq B_x\, \textnormal{, for all }1 \leq i \leq n
 \big\}\,{,}
\end{split}
\end{equation}
has an identical value and solution to~\ref{eq:exp0}.
\end{lemma}
\begin{proof}
Given the promise, the bounds on $x$ are redundant. This is also the case with the upper bounds on $t_{ij}$, since setting $x=\allzeros$ and $t_{ij}=\AA_{ij}$ yields a solution of value $s_{\AA}$. Therefore, setting $t^*_{ij} = \AA_{ij} e^{x^*_i -x^*_j}$, the value of $\langle \allones, t^*_{ij} \rangle$ must be at most $s_{\AA} + \langle  d, x^* \rangle \leq s_{\AA} + \Vert d \Vert_1 B_x < U$.
\end{proof}

The second step is to replace the hard constraints with appropriate barrier functions, whose value blows up when approaching the boundary of the feasible set $S$ (see~\cite{ben2001lectures, nesterov1998introductory} for more details). More precisely, we consider the barrier functions
\begin{equation}
\phi_{ij}(t, x) = -\ln(\ln t - (\ln \AA_{ij} + x_i - x_j))
\end{equation}
for all $(i,j) \in \supp(\AA)$, and
\begin{equation}
\psi(t,x) = -\sum_{(i,j) \in \supp(\AA) } \ln(3U - t_{ij}) - \sum_{1\leq i \leq n} (\ln(B_x - x_i)  
+ \ln(x_i + B_x) ) \,{.}
\end{equation}
The former blow up when $t$ approaches $\exp(\ln(\AA_{ij} + x_i -x_j))$ from above, and are  standard in exponential cone programming~\cite{ben2001lectures}. The barrier $\psi$ handles all the other inequality constraints. Very importantly, all these barrier functions are well behaved, in the sense that they satisfy a required property called \textit{self-concordance}. Since this property defines the number of iterations the method needs to execute, we highlight it below.

\begin{fact}\label{fact:sc}
The function $\xi(t,x) = \psi(t,x) + \sum_{(i,j) \in \supp(\AA)} \phi_{ij}(t,x)$ is an $O(m)$-self-concordant barrier for the set $S$ defined in~\ref{eq:opt1}.
\end{fact}

With the barrier function set up, the method has to solve a sequence of subproblems of the form 
\begin{equation}\label{eq:centralpath}
\min_{t,x} f_{\mu}(t,x) \quad \textnormal{where} \quad
f_{\mu}(t,x) = \mu \cdot (\langle \allones, t \rangle - \langle d, x \rangle) + \xi(t,x)
\end{equation}
while increasing $\mu$ until it becomes sufficiently large that the solution we produce is close to the optimum of the initial constrained problem. 

What is essential here is the number of iterations of the method, which depends mostly on the quality of the barrier function, and little on in initialization and accuracy to which we want to solve. More precisely, we apply the following theorem which follows from~\cite{nesterov1998introductory}, Theorems 4.2.9 and 4.2.11.
\begin{theorem}\label{thm:nesterovipm}
Given an initial point $v$ in the strict interior of $\mathcal{D}$ with a
$\nu$-self-concordant barrier $\xi$, the problem $\min_{v \in \mathcal{D}}
c^\top v$ can be solved to within $\epsilon$ additive error in $$O\left(
\sqrt{\nu} \ln \left ( 2 + \norm{\nabla \xi(v_0)}_{\nabla^2 \xi(v_0)^{-1}} + \nu \frac{\Vert c \Vert_{{\nabla^2 \xi(v_0)^{-1}}}}{\epsilon} \right) \right)$$ iterations, where $v_0$ is the minimizer over $\mathcal{D}$ of $\xi(v)$.
\end{theorem}

In what follows we bound the quantities involved in the  above statement. In order to have a bound on the number of iterations required for our cone program, we require lower bounding $\nabla^2 \xi(v_0)$. In order to do so, we lower bound the Hessian everywhere.

\begin{lemma}\label{lem:hesslb}
The Hessian $\nabla^2 \xi$ is lower bounded everywhere by the diagonal matrix with $\frac{1}{9U^2}$ on $t$ variables
and $\frac{2}{B_x^2}$ on on $x$ variables.
\end{lemma}
\begin{proof}
Since $\nabla^2 \xi = \nabla^2 \phi + \nabla^2 \psi \succeq \nabla^2 \psi$, and by calculation we see that $\nabla^2 \psi$ is diagonal, and
\begin{align*}
[\nabla^2 \psi(t,x)]_{t_{ij},t_{ij}} &= \frac{1}{(3U - t_{ij})^2} \geq \frac{1}{9U^2}\\
[\nabla^2 \psi(t,x)]_{x_{i}, x_{i}} &= \frac{1}{(B_x - x_i)^2} + \frac{1}{(B_x + x_i)^2} \geq \frac{2}{B_x^2}\,{.}\\
\end{align*}
Therefore we have that this gives a lower bound on  $\nabla^2 \psi$, and thus on $\HH$.
\end{proof}

We also show how to pick the initial point for our particular problem, which turns out to be a trivial task, since the only requirement is that it lies in the strict interior of $\mathcal{D}$. The more challenging part is upper bounding the $\nabla \xi$ at that point in Hessian inverse norm.
\begin{lemma}\label{lem:init}
The point $v = (t, x)$, where
 $t_{ij} =  2U$,
 $x = \allzeros$, 
 belongs to the strict interior of $S$. Furthermore, $\ln \Vert \nabla \xi(v) \Vert_{\nabla^2 \xi(v)^{-1}} = O(\ln(2 + m + B_x))$.
\end{lemma}
\begin{proof}
First we verify that the point belongs to the strict interior. As we set $x = \allzeros$, no constraint on $x$ is tight. $3U > t_{ij} = 2U > 0$.

For the second part, we may simply bound the contribution from each term of the barrier to each entry of the gradient.  The $t$ entries end up
bounded by at most $\frac{O(m)}{U}$, while the $x$ entries end up bounded by $O(m)+\frac{O(m)}{B_x}$, providing the claimed bound.
\end{proof}

\subsection{Implementing an Iteration of the Interior Point Method}  

The steps mentioned in the statement of Theorem~\ref{thm:fopt} consist only of standard Newton steps, i.e. minimizing a second order local approximation of the function $f_\mu(x)$. These steps are generally expensive, since they involve applying the inverse  of $\nabla^2 f_\mu$ to a vector. In our case, fortunately, we are able to exploit the structure of $f$ in order to do this in nearly linear time in the sparsity of $\nabla^2 f_\mu$.

Below we give a precise statement concerning our ability to solve linear systems involving the Hessian matrix.

\begin{theorem}\label{thm:hessian_solver} 
For any $\epsilon>0$, and any $\HH = \nabla^2 f_\mu(v) = \nabla^2 \xi(v)$, where $v=(t,x)$ is a point in the strict interior of the feasible region $S$ (see~\ref{eq:opt1}), and any vector $b \in \R^{m+n}$, one can, with high probability, compute  in $\tilde{O}(m \log\epsilon^{-1})$ 
time a vector $y$ such that $\Vert y-y^* \Vert_{\HH}\leq\epsilon\Vert y^{*}\Vert_{\HH}$,
where $y^{*}$ is the solution to $\HH y^* = b$.
\end{theorem}

In order to achieve this result, we leverage the power of Laplacian solvers.  From the algorithmic point of view, the crucial property of the
Laplacian  is that it is symmetric and \emph{diagonally dominant}. This
enables us to use fast approximate solvers for symmetric and diagonally
dominant linear systems. Namely, there is a long line of work \cite{SpielmanT04,KoutisMP10,KoutisMP11,KelnerOSZ13,CohenKMPPRX14,KyngLPSS16,KyngS16} that builds on an earlier work of Vaidya \cite{Vaidya90} and Spielman and Teng \cite{SpielmanT03}, that designed an SDD linear system solver. We employ as a black box the following theorem, which follows from~\cite{KoutisMP11}, and constructs an operator that approximates $\MM^+$.
\begin{theorem}
\label{thm:vanilla_SDD_solver} For any $\epsilon>0$, and any SDD  matrix $\MM \in \R^{n \times n}$ with $m$ nonzero entries, and any vector $b$ in the image of $\MM$, one can, with high probability, compute  in $\Otil(m\log\epsilon^{-1})$
time a vector $x$ such that $\Vert x-x^* \Vert_{\MM}\leq\epsilon\Vert x^* \Vert_{\MM}$, where $x^*$ is the solution of $\MM x^* = b$. Furthermore, a given choice of random bits produces a correct result for all $b$ simultaneously, and makes $x$ linear in $b$.
\end{theorem}

The result follows using this tool, and the following structural lemma, whose proof can be found in Appendix~\ref{sec:app3}.
\begin{lemma}\label{lem:factor}
The Hessian $\nabla^2 \xi$ has a factorization
$${\HH} = \UU \ms \UU^\top\,{,}$$ where $\ms$ is SDD, $\UU$ is lower triangular, and each of then has $O(m)$ nonzero entries. Furthermore,  this factorization can be computed in $O(m)$ time.
\end{lemma}

With this in hand, proving Theorem~\ref{thm:hessian_solver} is immediate. First note that, given $\widetilde{\ms}$, we can actually choose a $\widetilde{\ms}^{-1}$ satisfying the needed properties: a linear-operator based graph Laplacian solver, such as~\cite{KoutisMP11}. Having access to the linear operator $\widetilde{\ms}^{-1}$, we consider the error produced by applying the operator  $\UU^{-1}\widetilde{\ms}(\UU^\top)^{-1}$:
\begin{align*}
\Vert  (\UU^\top)^{-1}\widetilde{\ms}^{-1}\UU^{-1} b -   \HH^{-1}b  \Vert_{\HH} &=  
\Vert  (\UU^\top)^{-1}(\widetilde{\ms}^{-1} - \ms^{-1}) \UU^{-1} b \Vert_{\HH} 
= \Vert (\widetilde{\ms}^{-1} - \ms^{-1}) \UU^{-1} b 
\Vert_{\ms} \\
&\leq \epsilon \Vert  \ms^{-1} \UU^{-1} b \Vert_{\ms} 
= \epsilon \Vert (\UU^\top)^{-1} \ms^{-1} \UU^{-1} b \Vert_{\UU \ms \UU^\top}  
= \epsilon \Vert \HH^{-1} b \Vert_{\HH}\,{.}
\end{align*}

\subsection{Error Tolerance}

While the classical analysis of Newton's method used for iterations of interior point methods assumes exact computations, in our case the Laplacian solver we employ adds some error. We quickly show that this error does not hurt us, and as a matter of fact it is sufficient to solve these systems to constant accuracy.

First, we require understanding the guarantees of Newton's method, and its requirements.
\begin{fact}[Progress via Newton Steps]\label{fact:pvns}
Let $v$ be a point in the interior of the feasible region such that $$\Vert \nabla f_\mu (v)\Vert_{\HH_v^{-1}} \leq \frac{1}{4}\,{.}$$
 Then, applying one step of the interior point method consists of producing a new iterate 
$$v' = v - \HH_v^{-1} \nabla f_\mu(v)\, \quad \textnormal{which provably satisfies} \quad \Vert \nabla f_\mu (v')\Vert_{\HH_{v'}^{-1}} \leq \frac{1}{8}\,{.}$$
In order to make progress it is sufficient that $\Vert \nabla f_\mu(v')\Vert_{\HH_{v'}^{-1}} \leq \frac{1}{6}$.
\end{fact}

We can easily show that applying the inverse matrix with the solver guarantees we give in Theorem~\ref{thm:hessian_solver} is sufficient in order to make progress.

\begin{lemma}\label{lem:ns}
Let $v$ be a point in the interior of the feasible region such that $$\Vert \nabla f_\mu (v)\Vert_{\HH_v^{-1}} \leq \frac{1}{4}\,{.}$$
Letting  $v'' = v - \Delta$ such that
$$ \Vert  \Delta -  \HH_v^{-1} \nabla f_{\mu}(v) \Vert_{\HH_v} \leq \epsilon \Vert  \nabla f_\mu(v) \Vert_{\HH_v^{-1}}\,{,} $$ for $\epsilon \leq 0.1$, we get that
$$  \Vert \nabla f_{\mu}(v'')  \Vert_{ \HH_{v''}^{-1} } \leq \frac{1}{6} \,{.} $$
\end{lemma}
Since the proof is rather standard, we defer it to Appendix~\ref{sec:pfofns}.

\subsection{Putting Everything Together}
We can combine the results from this section in order to provide a proof for Theorem~\ref{thm:fopt}.
\begin{proof}[Proof of Theorem~\ref{thm:fopt}]
By combining 
Theorem~\ref{thm:nesterovipm}, along with Fact~\ref{fact:sc}, Lemma~\ref{lem:hesslb} and Lemma~\ref{lem:init}, we see that we can approximately minimize the function defined in Equation~\ref{eq:exp0} by performing 
$$O\left( \sqrt{m} \ln \left ( 2 + m + B_x + \frac{s_\AA}{\epsilon} + \frac{\Vert d \Vert_1}{\epsilon} \right ) \right)$$
iterations of the interior point method referenced in Theorem~\ref{thm:nesterovipm}.

From Theorem~\ref{thm:hessian_solver} and the iteration accuracy required 
by Fact~\ref{fact:pvns} and Lemma~\ref{lem:ns} we see that each iteration of the interior point method referenced in Theorem~\ref{thm:fopt} can be implemented in time $\Otil(m)$.  This yields the conclusion.
\end{proof}

%% file: appendix.tex
\section{Deferred Proofs}

\subsection{Proof of Theorem~\ref{thm:newtonoracle}}\label{sec:pfofnewtonoracle}
\begin{proof}[Proof of Theorem~\ref{thm:newtonoracle}]
The iteration we are going to implement is 
\begin{align}\label{it:boxoracle}
x_{i+1} &= x_{i} + \frac{1}{k} \cdot \oracle \left(\frac{e^2}{k^2} \nabla^2 f(x_i), \frac{1}{k} \nabla f(x) \right)\,{.}
\end{align}
Since $f$ is second-order robust with respect to $\ell_\infty$ by definition (see Definition~\ref{def:sor}), we know that within an $\ell_\infty$-ball centered at $x$ the function $f$ is lower and upper bounded by $f_L$ and $f_U$, respectively, where:
\begin{align*}
f_L(x') &= f(x) + \langle \nabla f(x), x' - x\rangle + \frac{1}{2 e^2}(x'-x)^\top \nabla^2 f(x) (x'-x)\,{,} \\
f_U(x') &= f(x) + \langle \nabla f(x), x' - x\rangle + \frac{e^2}{2}(x'-x)^\top \nabla^2 f(x) (x'-x)\,{.}
\end{align*}

Also, define $x_L$ and $x_U$ to be the minimizers of $f_L$ and $f_U$, respectively, over the $\ell_\infty$-ball of radius $\frac{1}{k}$ centered at $x$, i.e.
\begin{align*}
x_L = \argmin_{\Vert z-x \Vert_{\infty} \leq \frac{1}{k}} f_L(z)\quad { \textnormal{and}}\quad x_U = \argmin_{\Vert z-x \Vert_{\infty} \leq \frac{1}{k}} f_U(z)\,{.} 
\end{align*}

Next, we see how much $f$ decreases when we move from $x$ to $x' = x + \frac{1}{k} \Delta $, where $\Delta$ is obtained via the oracle call $\oracle \left(\frac{e^2}{k^2} \nabla^2 f(x_i), \frac{1}{k} \nabla f(x) \right)$.   We know from Definition~\ref{def:oracle} that 
\begin{equation}\label{eq:bd1}
\frac{1}{k} \langle \nabla f(x), \Delta \rangle + \frac{e^2}{2k^2} \Delta^\top \nabla^2 f(x) \Delta \leq \frac{1}{2}\left(\langle \nabla f(x), x_U - x \rangle + \frac{e^2}{2} (x_U - x)^\top \nabla^2 f(x) (x_U - x) \right) 
\end{equation}
as the function the oracle is approximately minimizing is precisely that.
Expanding $f_U \left ( x+\frac{1}{k}\Delta \right)$ we have
\begin{align*}
f_U(x')-f_U(x) &= f_U \left (x+\frac{1}{k}\Delta \right )-f_U(x) \\
&= \frac{1}{k} \langle \nabla f(x), \Delta \rangle + \frac{e^2}{2k^2} \Delta^\top \nabla^2 f(x) \Delta \\
&\leq \frac{1}{2} \left(\langle \nabla f(x), x_U - x \rangle + \frac{e^2}{2} (x_U - x)^\top \nabla^2 f(x) (x_U - x) \right) \\
&= \frac{1}{2} (f_U(x_U)-f_U(x))\,{.}
\end{align*}
 
Since $f(x) = f_U(x)$, we see that
\begin{equation}\label{eq:bd2}
f(x)-f(x') \geq f_U(x) - f_U(x') \geq \frac{1}{2}(f_U(x) - f_U(x_U))\,{.}
\end{equation}

Also, we have that
\begin{align*}
f_U(x_U) - f_U(x) &\leq  f_U \left ( x+\frac{x_L-x}{e^4} \right ) - f_U(x)\\
&= \left \langle \nabla f(x), \frac{x_L-x}{e^4} \right \rangle + \frac{1}{2 e^4}\left ( \frac{x_L-x}{e^4} \right )^\top \nabla^2 f(x) \left ( \frac{x_L-x}{e^4} \right ) \\
&= \frac{1}{e^4} \left ( \langle \nabla f(x), x_L-x \rangle + \frac{1}{2 e^2} (x_L-x)^\top \nabla^2 f(x) (x_L-x) \right ) \\
&= \frac{1}{e^4} (f_L(x_L) - f_L(x))\,{.}
\end{align*}
Combining this with Equation~\ref{eq:bd2} gives
\begin{equation}\label{eq:bd3}
f(x)-f(x') \geq \frac{1}{2e^4} (f_L(x) - f_L(x_U))\,{.}
\end{equation}

Finally, we show that this amount of progress is comparable to that achievable by making a
large step towards $x^*$. More precisely, we have from the $R_\infty$ condition that
$\Vert x-x^* \Vert_\infty \leq R_\infty$. Thus, letting $\hat{x} = x + \frac{1}{\max(k R_\infty, 1)} (x^* - x)$, we have that $\Vert \hat{x} - x^*\Vert_\infty \leq \frac{1}{k}$. Therefore, $f_L(
\hat{x}) \geq f_L(x_L)$, since $x_L$ was a minimizer of $f_L$ over the $\ell_\infty$-ball of radius $\frac{1}{k}$ around $x$. Also, since $f_L$ lower bounds $f$ over this $\ell_\infty$-ball, 
\begin{equation}\label{eq:bd4}
f_L(x_L) \leq f_L(\hat{x}) \leq f(\hat{x}).
\end{equation}
Combining Equations~\ref{eq:bd3} and ~\ref{eq:bd4}, we see that
\begin{align*}
f(x)-f(x') \geq \frac{1}{2e^4} (f(x) - f_L(x_L)) \geq \frac{1}{2e^4 k} (f(x) - f(\hat{x})) \geq \frac{1}{2 e^4 \max(k R_\infty, 1)} (f(x)-f(x^*))\,{,}
\end{align*}
where the last inequality follows from convexity.
This implies that at every iteration $f(x) - f(x^*)$ is decreased by a factor of
$(1-\Omega(1/(kR_\infty+1)))$,
implying that after $$T = O\left((kR_\infty+1) \cdot \log\left( \frac{f(x_0)-f(x^*)}{\epsilon}  \right)\right) $$ iterations, we have that $f(x_T) - f(x^*) \leq \epsilon$.
\end{proof}

\subsection{Proof of Lemma~\ref{lem:scaling_accuracy}}\label{sec:pfofscaling_accuracy}
\begin{proof}[Proof of Lemma~\ref{lem:scaling_accuracy}]
Suppose that, without loss of generality, the $i^{th}$ row of $\MM$ has a very large violation of the scaling constraint: letting $\gamma := (r_\MM)_i - r_i$ we have $\abs{\gamma} \geq \epsilon/\sqrt{2n}$. 

In order to improve the solution, we can make an update to the corresponding coordinate of $x$ which makes the largest possible improvement in function value. More precisely by setting $x_i' = x_i + \delta$, and $x_j' = x_j$ whenever $j \neq i$, we have that
\begin{align*}
f(x) - f(x') = (r_\MM)_i (1-e^\delta) + r_i \delta\,{.}
\end{align*}
Optimizing for the largest possible decrease, we set $\delta = \ln(r_i/(r_\MM)_i)$ which
shows that we can decrease $f$ by
\begin{align*}
f(x) - f(x') = (r_\MM)_i - r_i - r_i\ln \left(1+\frac{(r_\MM)_i-r_i}{r_i}\right) =
r_i \left( \frac{\gamma}{r_i} -\ln \left(1+\frac{\gamma}{r_i}\right)\right)\,{.}
\end{align*}

Since we have $\gamma/r_i \geq -1$, we can lower bound the improvement by
 $$f(x) - f(x') \geq r_i \cdot \frac{(\gamma/r_i)^2}{4}  \geq \frac{1}{r_i} \cdot \frac{\epsilon^2}{2n}\,{,}$$
 whenever $\gamma/r_i \leq 1.62$,
 and by
 $$f(x) - f(x') \geq r_i \cdot \frac{\gamma /r_i}{3} = \frac{\gamma}{3} \geq \frac{\epsilon}{3\sqrt{2n}}  \,{,}$$
 whenever $\gamma  / r_i > 1.62$.
 
 Since by assumption $\|r\|_\infty \leq 1$, this change improves function value by at least $\min\{  \epsilon^2 /(2n), \epsilon/(3\sqrt{2n})\}$, which contradicts the fact that $f(x) - f^*\leq \epsilon^2 / 3n$. Therefore all rows and columns are within $\epsilon/\sqrt{2n}$ away from being correctly scaled. Hence this is a $\epsilon$-$(r,c)$ scaling.

\end{proof}

\subsection{Proof of Lemma~\ref{lem:scalreg}}\label{sec:pfofscalreg}

\begin{proof}
The proof is similar to the one for Lemma~\ref{lem:balreg}. The first point holds by the same argument. 

For the third part, all $(x,y)$ for which $\widetilde{f}(x,y) \leq \widetilde{f}(0,0)$  must satisfy:
$$\frac{\epsilon^2}{36n^2e^B} (e^{x_i} + e^{-x_i}) \leq f(0,0) + \frac{\epsilon^2}{36n^2e^B} \cdot 4 = \uA + \frac{\epsilon^2}{9ne^B}\,{,}$$
and similarly for $y$. Therefore
$$\abs{x_i} \leq \ln\left(  \frac{36n^2e^B}{\epsilon^2}\uA + 4 \right) = O( B \ln (n\uA\epsilon^{-1}) )\,{,}$$ and similarly for $y_i$. 

The second part follows by the nonnegativity of the regularizer and the observation that 
$\widetilde{f}(z^*_\epsilon) \leq f(z^*_\epsilon) + \epsilon^2/(36n^2e^B)\cdot n \cdot 4e^B
\leq f(z^*_\epsilon) + \epsilon^2/(9n)$. By the third property we know that the level set in
bounded and thus $\widetilde f$ attains its minimum, and that minimum can only be better that
$x^*$, which concludes the proof.
\end{proof}

\subsection{Proof of Theorems~\ref{thm:scaling} and~\ref{thm:escaling}}\label{sec:mainscal}
\begin{proof}[Proof of Theorem~\ref{thm:escaling}]
By Lemma~\ref{lem:scaling_accuracy} and Lemma~\ref{lem:scalreg} we get that in order to obtain a $2\epsilon$-approximate scaling, it is sufficient to minimize $\widetilde{f}$ up to $\epsilon^2/(2n)$ additive error.
Furthermore, from Lemma~\ref{lem:scalreg} we get that the $R_\infty$ bound required for Theorem~\ref{thm:sor_sdd} is $R_\infty = O(B\ln(n\uA \epsilon^{-1}))$.
Finally, since $f(0) = \uA$ and $f(z^*)\geq O(n + B)$, we see that, initializing at $(x_0, y_0) = (0,0)$, the total
running time of the method is upper bounded by
$$\Otil\left(  m B \ln^2 (\uA\epsilon^{-1})   \right)\,{.}$$
\end{proof}

\begin{proof}[Proof of Theorem~\ref{thm:scaling}]
We can directly prove this theorem by applying Theorem~\ref{thm:escaling} to the optimal
solution promised. Since $z^*$ exactly $(r,c)$-scales $\AA$, we know that it must be a
minimizer of $f$ and thus $f(z^*) = f^*$. 
Moreover, by definition we have the bound,
$\|z^*\| = B \leq \ln (\kappa(\UU^*_{\epsilon}) + \kappa(\VV^*_\epsilon))$,
which concludes that proof.
\end{proof}

\subsection{Proof of Lemma~\ref{lem:almostdss}}\label{sec:pfofalmostdss}
\begin{proof}
By Lemma~\ref{lem:scalcond}, we know that any almost scalable
matrix can be written as a block lower triangular matrix, whose diagonal blocks are exactly
scalable.
By Lemma~\ref{lem:kkn}, every such block can be scaled to doubly stochastic, using factors
with a ratio at most $O(n_i\log(1/\lA))$, where $n_i$ is the number of vertices in block $i$.

The infimum of the function value is exactly the sum of the function values for the diagonal
block problems, since the contribution of the entries below the diagonal can be made
arbitrarily close to $0$. We observe that it suffices to ensure that the contribution of each
such edge is at most $\eps^2/3n^3$, since then the total contribution will be at most
$\eps^2/3n$ which is the additive error we can tolerate. Scaling the off-diagonal entries can
be done in a very simple way. For any block, we can scale all the columns down by a fixed
amount and all the columns up by the same amount. This will not affect the contribution of
the block's entries to the function and will only decrease the contribution of all the
off-diagonal blocks in the same columns. By choosing the ratio between any two consecutive
blocks to be $\log(n^3\uA/\eps^2)$, we can ensure that the entries contained in the
interesection of the rows and columns of these blocks contribute less than $\eps^2/3n^3$
each. That ratio between any two factors of this new scaling is at most
$$O\left(n\log(n\uA/\eps) + \sum_i n_i\log(1/\lA)\right)\leq O(n\log(n\wA/\eps)).$$
\end{proof}

\subsection{Proof of Lemma~\ref{lem:balancing_requirement}}\label{sec:pfofbalreq}
\begin{proof}[Proof of Lemma~\ref{lem:balancing_requirement}]
First we observe that since the Hessian of $f$ is SDD, it is spectrally upper bounded by two times its
diagonal and therefore by the identity matrix multiplied by twice the trace, that is
$\nabla^2f(x)\preceq 2\cdot\tr(\nabla^2f(x))\cdot\mI$. Since, by construction,
$\tr(\nabla^2 f(x)) = \sum_i
(r_\MM + c_\MM)_i = 2f(x)$, we have that
 $$\nabla^2 f(x) \preceq 4 f(x) \mI\,{.}$$
Therefore for any $y$ with $f(y)\leq f(x)$, we have that for some $t\in[0,1]:$
\begin{align*}
f(y)  &= f(x) + \langle \nabla f(x), y - x \rangle + \frac{1}{2} \cdot (y - x)^\top
          \nabla^2 f(x+t(y - x)) (y - x) \\
&\leq f(x) + \langle \nabla f(x), y - x \rangle + 2\|x-y\|_2^2 \cdot f(x+t(y-x))   \\
&\leq f(x) + \langle \nabla f(x), y - x \rangle + 2\|x-y\|_2^2 \cdot f(x).
\end{align*}
It is straightforward to reason that
\begin{align*}
f_* &= \inf_y f(y) 
    \leq f(x) + \min_y \left\{\langle \nabla f(x), y - x \rangle
        + 2\|x-y\|_2^2 f(x)\right\}
    = f(x) - \frac{\|\nabla f(x)\|_2^2}{8f(x)},
\end{align*}
and thus,
$$\frac{\|\nabla f(x)\|_2}{f(x)} \leq \sqrt{\frac{8(f(x) - f_*)}{f(x)}} \,{.}$$
Finally, we lower bound $f(x)$. Since the matrix can be balanced, its corresponding graph is strongly connected. Therefore it contains a cycle, and thus some edge $(i,j)$ satisfies $e^{x_i - x_j} \geq 1$. Hence $f(x) \geq \AA_{ij} \geq \lA$. 
Plugging in this lower bound, we get that $$f(x) - f_* \leq \frac{\epsilon^2 \lA}{8} \leq
\frac{\epsilon^2}{8} f_*\,{.}$$
Hence
$$\frac{\| \nabla f(x)\|_2}{f(x)} \leq  \sqrt{\frac{\eps^2 f_*}{f_*}}\leq \eps\,{,}$$
which is equivalent to the fact that $\mdiag(\exp(x))$ yields an $\epsilon$-balancing for $\AA$.
We note that a similar bound also follows from~\cite{ORY17}, using a different argument.
 \end{proof}

\subsection{Proof of Theorems \ref{thm:balancing} and
\ref{thm:ebalancing}}\label{sec:mainbal}
\begin{proof}[Proof of Theorem~\ref{thm:ebalancing}]
By Lemma~\ref{lem:balancing_requirement} and Lemma~\ref{lem:balreg} we get that optimizing
$\widetilde f$ up to an additive error of $\eps^2\AA_{\min}/24$, suffices to get an
$\eps$-balancing of the matrix.

Furthermore, from Lemma~\ref{lem:balreg} we see that the $R_\infty$ bound required for
Theorem~\ref{thm:sor_sdd} is $R_\infty = O(B \log (n\wA \epsilon^{-1}))$. Finally, using the
fact that $f(0) = \uA$, we see that, initializing at $x_0 = 0$, the total running time of the
method is
$$\Otil(m B \log^2(\wA \epsilon^{-1}) )\,{.}$$
\end{proof}

\begin{proof}[Proof of Theorem~\ref{thm:balancing}]
Having proved Theorem~\ref{thm:ebalancing}, this theorem is a simple corollary. Consider
$x^*$ to be the vector such that $\DD^* = \mdiag(\exp(x^*))$. That implies that $\nabla
f(x^*)=0$ and therefore (by the convexity of $f$) $x^*$ is a minimizer of $f$ implying that
$f(x^*) = f^*$. Moreover, $B = \max_i \abs{\log\DD^*_{ii}} = O(\log \kappa(\DD^*))$,
which concludes that proof by applying Theorem~\ref{thm:ebalancing}.
\end{proof}

\subsection{Proof of Lemma~\ref{lem:baldiam}}\label{sec:pfofbaldiam}
\begin{proof}
Consider the optimal solution $x$ to the optimization problem described in Equation~\ref{eq:fbal}, for which we know that $\DD^* = \mdiag(\exp(x))$ via Lemma~\ref{lem:balancing_requirement}. Since this is a minimizer, we know that
\begin{align*}
\sum_{(i,j) \in \supp(\AA)} \AA_{ij} e^{x_i - x_j} = f(x) \leq f(0) = \uA\,{.}
\end{align*}
Therefore, for any $(i,j) \in \supp(\AA)$, one has that 
$$ x_i - x_j \leq \ln (\uA / \AA_{ij}) \leq \ln \wA\,{.} $$
Since there is a directed path of length at most $k$ from any vertex to any other, we get
that
$$\ln \kappa(\DD^*) = \max_i x_i - \min_j x_j = O(k \ln \wA)\,{.}$$
\end{proof}


\subsection{Proof of Lemma~\ref{lem:ns}}\label{sec:pfofns}\label{sec:ipmproofs}

\begin{proof}
The first part is a standard property of Newton's method applied to self-concordant functions. We refer the reader to~\cite{ben2001lectures} for details. 

What we want to prove is that Newton's is robust to errors in the solution to the linear system involving the Hessian.
Indeed, first we see that the Hessian at $v''$ approximates the one at $v'$. To simplify notation, we write $\HH_v = \nabla^2 f(v)$. Since $f_\mu$ is self-concordant, we have that
\begin{equation}\label{eq:move0}
\HH_{v'} \cdot (1-  \Vert v'-v'' \Vert_{\HH_{v'}} )^2 \preceq \HH_{v''} \preceq \HH_{v'} \cdot \frac{1}{(  1 - \Vert v'-v'' \Vert_{\HH_{v'}}  )^2} 
\end{equation}
and similarly
\begin{equation}\label{eq:move}
\HH_{v} \cdot (1-  \Vert v-v' \Vert_{\HH_{v'}} )^2 \preceq \HH_{v'} \preceq \HH_{v} \cdot \frac{1}{(  1 - \Vert v-v' \Vert_{\HH_{v}}  )^2} 
\end{equation}
the latter of which can be written equivalently as
\begin{equation*}
\HH_{v} \cdot (1-  \Vert \nabla f_\mu(v) \Vert_{\HH_{v}^{-1}} )^2 \preceq \HH_{v'} \preceq \HH_{v} \cdot \frac{1}{(  1 - \Vert \nabla f_\mu(v) \Vert_{\HH_{v}^{-1}}  )^2} \,{,}
\end{equation*}
so
\begin{equation}\label{eq:hessappx}
\HH_v \cdot \frac{9}{16} \preceq \HH_{v'} \preceq \HH_v \cdot \frac{16}{9}\,{.}
\end{equation}

The error guarantee on $v''$ equivalently gives us  that
\begin{equation}
\Vert v'-v'' \Vert_{\HH_v} \leq \epsilon \Vert v'-v \Vert_{\HH_v}\,{,}
\end{equation}
so combining with~\ref{eq:hessappx} we obtain that
\begin{equation}\label{eq:normmess}
\Vert v' - v'' \Vert_{\HH_{v'}} \leq \frac{4}{3} \Vert v'-v'' \Vert_{\HH_v} \leq \frac{4}{3}\epsilon \Vert v'-v \Vert_{\HH_v} 
=\frac{4}{3}\epsilon \Vert  \nabla f_\mu(v) \Vert_{\HH_v^{-1}}  \leq \frac{\epsilon}{3}\,{.}
\end{equation}

Also, since for any $z$
\begin{align}
\nabla f_\mu(z + w) = \nabla f_\mu(z) + \int_0^1 \HH_{z+t w} w \cdot dt \nonumber
\end{align}
we get, by applying triangle inequality and~\ref{eq:move}, that
\begin{equation}\label{eq:grad_change}
\Vert \nabla f_\mu(z + w) \Vert_{\HH_z^{-1}} \leq \Vert \nabla f_\mu(z) \Vert_{\HH_z^{-1}} +  \frac{1}{1-\Vert w \Vert_{\HH_z} }  \left\Vert  \HH_z w \right\Vert_{\HH_z^{-1}} \,{.}
\end{equation}

Therefore, using~\ref{eq:move0} and~\ref{eq:grad_change} where we substitute $v'$ for $z$:
\begin{align*}
\Vert \nabla f_\mu(v'') \Vert_{\HH_{v''}^{-1}} 
&\leq  \Vert  \nabla f_\mu(v'') \Vert_{\HH_{v'}^{-1}} \cdot
 \frac{1}{1-\Vert v'-v'' \Vert_{\HH_{v'}}}   \\
&\leq \left( \Vert  \nabla f_\mu(v') \Vert_{\HH_{v'}^{-1}} + 
\frac{\left\Vert  \HH_{v'} (v'-v'') \right\Vert_{\HH_{v'}^{-1}} }{1-\Vert v'-v'' \Vert_{\HH_{v'}}}  
 \right) \cdot
 \frac{1}{1-\Vert v'-v'' \Vert_{\HH_{v'}^{-1}}}   \\
 &\leq \left( \frac{1}{8} + \frac{\epsilon / 3}{1-\epsilon /3 } \right)\cdot \frac{1}{1-\epsilon/3}\\
 &\leq \frac{1}{6}\,{.}
\end{align*}

\end{proof}

\subsection{Proof of Lemma~\ref{lem:factor}}\label{sec:app3}
\begin{proof}
First we note that the nonzero submatrix of
 $\nabla^2 \phi_{ij}$ is (where rows/columns correspond to $x_i$, $x_j$, $t_{ij}$, in this order):
\begin{align*}
\nabla^2 \phi_{ij}(x_i, x_j, t_{ij}) &= 
\begin{bmatrix} 
 \alpha & -\alpha & -\alpha/t_{ij} \\
 -\alpha & \alpha & \alpha/t_{ij} \\
 -\alpha/t_{ij} & \alpha/t_{ij} & \beta/t_{ij}^2 \\
\end{bmatrix}
\end{align*}
such that $$\alpha = \frac{1}{( \ln t_{ij} - (\ln \AA_{ij} + x_i -x_j ) )^2}\quad\textnormal{and}\quad \beta = \alpha + \sqrt{\alpha} + 1\,{.} $$

Furthermore, this submatrix can be factored, by Schur complementing the last row and column, as:
\begin{align*}
\nabla^2{\phi_{ij}}(x_i,x_j,t_{ij})
&=
\begin{bmatrix} 
1 & 0 & -\frac{\alpha}{\beta} \\
0 & 1 & \frac{\alpha}{\beta} \\
0 & 0 & 1
\end{bmatrix}
\begin{bmatrix} 
\alpha^2(1-
\frac{1}{\beta}) & -\alpha^2(1-
\frac{1}{\beta}) & 0 \\
-\alpha^2(1-
\frac{1}{\beta}) & \alpha^2(1-
\frac{1}{\beta}) & 0 \\
0 & 0 & \beta/t_{ij} 
\end{bmatrix}
\begin{bmatrix} 
1 & 0 & 0\\
0 & 1 & 0 \\
-\frac{\alpha}{\beta}  & \frac{\alpha}{\beta}  & 1
\end{bmatrix}\,{,}
\end{align*}
and thus one can easily notice that the Schur complement is SDD.

Furthermore, since $\psi$ is a standard logarithmic barrier, its Hessian is a diagonal matrix with nonnegative entries. Therefore, we can split the contribution of the diagonal matrix $\nabla^2 \psi(x,t)$ into pieces $\DD_{ij}$ which contains nonzeroes only at $x_i$, $x_j$ and $t_{ij}$. In other words, we can write ${\HH} = \sum_{(i,j)\in\supp(\AA)} (\nabla^2 \phi_{ij}(x,t) + \DD_{ij})$.
Since the Schur complement of $t_{ij}$ of the matrix $\nabla^2 \phi(x,t)$ is SDD, we also have that the Schur complement of $t_{ij}$ of the matrix $\nabla^2 \phi(x,t) + \DD_{ij}$ is SDD, so the matrix can also be factored similarly to the factoring above, and all of these factorizations can be computed in overall $O(m+n) = O(m)$ time.

Finally, since each of these factorizations is computed by Schur complementing a unique $t_{ij}$, which is nonzero in a single matrix, we see that the Schur complement of the block ${\HH}(t)$ of the matrix ${\HH}$ is equal to the sum of Schur complements of the block $t_{ij}$ of the matrices $\nabla^2 \phi_{ij} (x,t) + \DD_{ij}$. This holds similarly, for the corresponding lower and upper diagonal matrices, which yields the desired factorization simply by summing up.
\end{proof}